\pdfoutput=1
\newif\ifcgta
 \cgtafalse
\newif\ifarxiv    \ifcgta \arxivfalse \else \arxivtrue \fi

\ifcgta
\documentclass{elsarticle}
\biboptions{numbers,sort&compress}
\else
\documentclass[11pt,a4paper]{article}
\usepackage[numbers,sort&compress]{natbib}
\fi
\usepackage{amsmath,amssymb,amsthm,graphicx,fullpage,bbding,xspace,url,wrapfig}
\usepackage{microtype}
\newtheorem{theorem}{Theorem}[section]\newtheorem{lemma}[theorem]{Lemma}\newtheorem{proposition}[theorem]{Proposition}\newtheorem{obs}[theorem]{Observation}\newtheorem{cor}[theorem]{Corollary}
\theoremstyle{definition}\newtheorem{definition}[theorem]{Definition}\theoremstyle{remark}\newtheorem{remark}{Remark}
\newcommand{\e}{\emph}\renewcommand{\cal}[1]{\ensuremath{\mathcal{{#1}}}\xspace}
\newcommand\cori{$C$-oriented\xspace}
\renewcommand{\c}{{\textrm{\ScissorLeft}}}
\newcommand{\opt}{\ensuremath{\mathsf{opt}}\xspace}\newcommand{\apx}{\ensuremath{\mathsf{apx}}\xspace}\newcommand{\apxm}{\ensuremath{\mathsf{apx}_m}\xspace}\newcommand{\apxl}{\ensuremath{\mathsf{apx}_1}\xspace}
\begin{document}
\date{}\title{Minimum-Link Paths Revisited}
\ifarxiv
	\author{Joseph S. B. Mitchell\thanks{Department of Applied Mathematics and Statistics, Stony Brook University. \texttt{joseph.mitchell@stonybrook.edu}} \and Valentin Polishchuk\thanks{Helsinki Institute for Information Technology, Department of Computer Science, University of Helsinki. \texttt{firstname.lastname@helsinki.fi}} \and Mikko Sysikaski\footnotemark[2]}\maketitle
\else
	\graphicspath{{cgta/}}
	\begin{frontmatter}
	\author{Joseph S. B. Mitchell \fnref{stonybrook}}\ead{joseph.mitchell@stonybrook.edu}
	\author{Valentin Polishchuk\fnref{helsinki}}\ead{valentin.polishchuk@helsinki.fi}
	\author{Mikko Sysikaski\fnref{helsinki}}\ead{mikko.sysikaski@helsinki.fi}
	\fntext[stonybrook]{Department of Applied Mathematics and Statistics, Stony Brook University}
	\fntext[helsinki]{Helsinki Institute for Information Technology, Department of Computer Science, University of Helsinki}
\fi

\begin{abstract}A path or a polygonal domain is \e{$C$-oriented} if the orientations of its edges belong to a set of $C$ given orientations; this is a generalization of the notable rectilinear case ($C=2$). We study exact and approximation algorithms for minimum-link $C$-oriented paths and paths with unrestricted orientations, both in $C$-oriented and in general domains.

Our two main algorithms are as follows:
\begin{description}
\item A subquadratic-time algorithm with a non-trivial approximation guarantee for general (unrestricted-orientation) minimum-link paths in general domains.
\item An algorithm to find a minimum-link $C$-oriented path in a $C$-oriented domain. Our algorithm is simpler and more time-space efficient than the prior algorithm.
\end{description}

We also obtain several related results:
\begin{itemize}
\item 3SUM-hardness of determining the link distance with unrestricted orientations (even in a rectilinear domain). 
\item An optimal algorithm for finding a minimum-link rectilinear path in a rectilinear domain. The algorithm and its analysis are simpler than the existing ones.
\item An extension of our methods to find a $C$-oriented minimum-link path in a general (not necessarily $C$-oriented) domain.
\item A more efficient algorithm to compute a 2-approximate $C$-oriented minimum-link path.
\item A notion of ``robust'' paths. We show how minimum-link $C$-oriented paths approximate the robust paths with unrestricted orientations to within an additive error of 1.
\end{itemize}

\end{abstract}
\ifcgta\end{frontmatter}\fi

\section{Introduction}Minimum-link problems arise in motion planning with turn costs, in line simplification, guarding applications, VLSI, wireless communication, and other areas. 
An instance of the problem is specified by an $n$-vertex polygonal domain $P$ with $h$ holes, and two points $s,t\in P$; the goal is to find an $s\-t$ path with the fewest edges (links). In the query version of the problem, the goal is to build a data structure (link distance map) to efficiently answer link distance queries with~$s$ fixed.

The algorithm of Mitchell, Rote and Woeginger \cite{mrw} computes a minimum-link path in $O(n^2\alpha^2(n)\log n)$ time, where $\alpha$ is the inverse Ackermann function. It was believed that a faster algorithm is possible (e.g., in \cite[p.~263]{dn} the result of \cite{mrw} is called ``suboptimal''). Nevertheless, the only previously known lower bound, also due to \cite{mrw}, was $\Omega(n\log n)$. The same bounds for the \e{rectilinear} case are given in \cite{dn,lyw}. 
Also, no approximation algorithm was previously known.

In this paper (Section~\ref{sec:unrestricted}) we give a subquadratic-time $O(\sqrt{h})$-approximation algorithm for the minimum-link path problem. We also observe (Theorem~\ref{thm:hard}) that finding the exact solution is 3SUM-hard; this answers a question from the survey \cite{SPsurvey} 
and Problem 22 in The Open Problems Project \cite{topp}.

Our 3SUM-hardness proof suggests that the problem's complexity stems from allowing the path edges to go in arbitrary directions. This---along with practical considerations---motivates the restricted, \e{$C$-oriented} setting \cite{aos,www,guting,gutingOttmann,hs,rw,neyer} (Fig.~\ref{fig:coriExample}) in which orientations of path edges come from a fixed set $C$ of directions. (Abusing notation, we use $C$ to denote also the cardinality of the set $C$.) Adegeest, Overmars and Snoyeink \cite{aos} presented two algorithms for finding minimum-link \cori paths in \cori domains -- one running in $O(C^2n\log n)$ time and space, the other in $O(C^2n\log^2n)$ time and $O(C^2n)$ space.
\begin{figure}\centering\includegraphics[width=.3\textwidth]{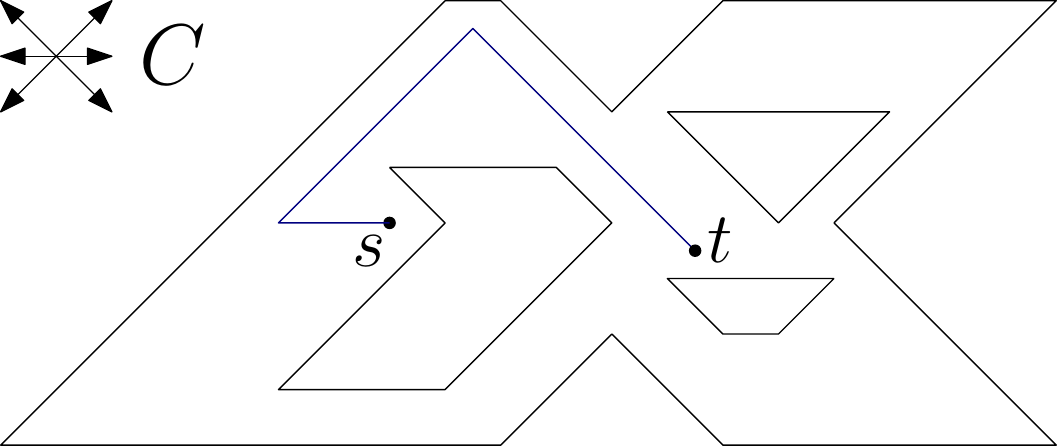}\caption{A \cori domain and a minimum-link \cori path in it.}\label{fig:coriExample}\end{figure}

In Section~\ref{sec:cori} we present an $O(C^2n\log n)$-time $O(Cn)$-space algorithm, slightly improving on both algorithms from~\cite{aos}.\footnote{The algorithm was also presented at WADS 2011~\cite{wads11}.} As a by-product, in Section~\ref{sec:rect}, we reestablish the optimal time and space bounds claimed in \cite{ylwBends} for computing a minimum-link \e{rectilinear} path amidst rectilinear obstacles. Unlike the earlier papers on the rectilinear case, we use only elementary data structures, which simplifies the algorithm and its analysis. We also show how to find a \cori path in a general domain (Section~\ref{sec:coriinarb}), give an $O(Cn\log n)$-time $O(n)$-space 2-approximation algorithm for \cori paths (Section~\ref{sec:cori}), and investigate in what sense \cori paths can approximate minimum-link paths with unrestricted orientation (Section~\ref{sec:robust}).

All of our algorithms not only find minimum-link paths but also build, within the same time and space bounds, the corresponding link distance maps---exact or approximate. For instance, using our algorithms, one can construct approximate (additive or multiplicative) maps for general minimum-link paths in general domains in subquadratic time and linear space. This is in contrast with the exact link distance maps, which may have quartic complexity \cite{so}.

\section{Paths with unrestricted orientations}\label{sec:unrestricted}The 3SUM-hardness of finding a minimum-link path can be seen easily, as we now observe. Start from an instance of the 3SUM-hard problem GeomBase considered in~\cite{go95}: Given a set $S$ of points lying on 3 parallel lines $l_1,l_2,l_3$, do there exist 3 points from $S$ lying on a line $l\notin\{l_1,l_2,l_3\}$? Construct an instance of the minimum-link path problem as follows (Fig.~\ref{fig:reduction}, left): $l_1,l_2,l_3$ become obstacles, and each point $p\in S$ is a gap punched in the obstacle. The $s\textrm-t$ link distance is 3 if and only if there exist 3 collinear gaps $p_i,i=1,2,3$, such that $p_i\in l_i$.

We thus obtain:\begin{theorem}\label{thm:hard}Determining the link distance, for paths with unrestricted orientations, between two points of a polygonal domain with holes is 3SUM-hard.  In particular, it is 3SUM-hard to decide if there exists a 3-link path between two points in a rectilinear domain.\end{theorem}


\begin{remark}
One can decide if the link distance between points $s$ and $t$ is 1 in time $O(n)$ (just test the segment $st$ for intersection with each edge of the domain).  One can test if the $s\textrm-t$ link distance is $\leq 2$ in time $O(n\log n)$ (just compute the visibility polygons with respect to $s$ and $t$, in time $O(n\log n)$, and test them for intersection, in time $O(n)$).  One can test if the $s\textrm-t$ link distance is $\leq 3$ in time $O(n^2)$ (assuming the visibility polygons with respect to $s$ and $t$ are disjoint, construct the visibility graph within the domain obtained by subtracting the two visibility polygons and check if there exists an (extended) visibility graph edge with endpoints on each of the two visibility polygons).
\end{remark}

Several corollaries are immediate: finding a $4/3-\varepsilon$ multiplicative approximation is 3SUM-hard; there is little hope to design an output-sensitive algorithm that would spend $o(n^2)$ time per link in the optimal path; computing an additive-1 approximation is 3SUM-hard; obstacles having few orientations of edges do not make the problem simpler, etc.
\begin{figure}\centering\includegraphics{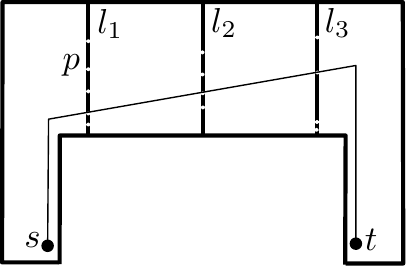}\hfil\includegraphics{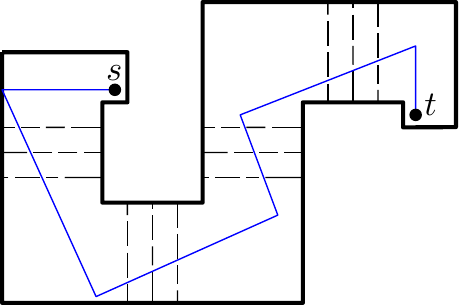}\hfill\caption{Left: The reduction. Right: Hardness of approximation.}\label{fig:reduction}\end{figure}

The proof can be strengthened to show that obtaining an $O(1)$ additive approximation is 3SUM-hard, and that obtaining a $(2-\varepsilon)$ multiplicative approximation is equally hard. To see this, let $k=O(1)$ be an arbitrary integer. Take an instance of GeomBase of size $n/k$, and make $k$ copies of it. Place the copies in a $k$-channel (``zig-zag'') corridor with $k$ channels, with one copy per channel (Fig.~\ref{fig:reduction}, right). There is a path utilizing a single link per channel (i.e., one link per copy of the instance) if and only if the GeomBase instance is feasible, otherwise 2 links per copy are needed. That is, if GeomBase is feasible, the $s\-t$ path will have $2+k$ links, otherwise it will have $2+2k$ links. Distinguishing between the two cases is at least as hard as solving the GeomBase instance of size $\Theta(n/k)=\Theta(n)$.

Hence, we have:
\begin{cor}\label{cor:2hard}Obtaining a $(2-\varepsilon)$-approximation to the link distance within a rectilinear domain is 3SUM-hard.\end{cor}

\subsection{$O(\sqrt{h})$-approximation algorithm}\label{sec:root_h}We start with an intuitive description of our approach; the technical details follow. We first review the solution to the minimum-link path problem in a \e{simple} polygon ($h=0$).

When the domain $P$ has no holes, a minimum-link path in it can be found in linear time \cite{suri,ghosh,hs} (assuming a real RAM; see \cite{kahan} for bit complexity issues). The general technique is to perform a ``staged illumination'' of the polygon \cite[Sections~26.4, 27.3]{handbook}: At the first stage, place the light source at $s$, and illuminate the visibility polygon of $s$. At the beginning of any subsequent stage, the boundary between the lit and the dark portions of $P$ is a set of ``windows''. Each window $w$ is a chord cutting out a dark portion $P_w\subset P$ (Fig.~\ref{fig:simple}). The crux of the linear-time algorithm is that $P_w$ is unique to $w$: the portion of $P_w$ illuminated at the next stage is exactly what is seen from $w$ (even if a point $p\in P_w$ is seen from another window $w'\ne w$, we do not care about shining light from $w'$ into $P_w$). At any stage, the illumination is guided by the dual graph of a triangulation of $P$. Since the dual is a tree, any triangle is lit only once (triangles intersected by windows are first lit only partially, and are fully lit at the next stage).
\begin{figure}\centering
\includegraphics[page=1]{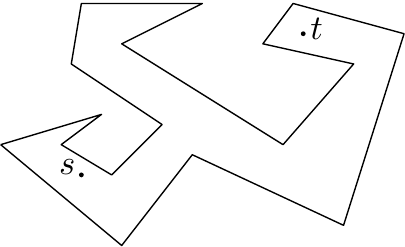}\includegraphics[page=2]{simple}\includegraphics[page=3]{simple}\vskip 10pt
\includegraphics[page=4]{simple}\includegraphics[page=5]{simple}\includegraphics[page=6]{simple}
\caption{Staged illumination in a simple polygon. The windows are yellow.}\label{fig:simple}\end{figure}

To obtain an $O(\sqrt h)$-approximation for the case in which $P$ is a polygonal domain with holes, we employ the staged illumination idea as follows (Fig.~\ref{fig:overview}). We bridge the holes to the outer boundary of $P$, thereby obtaining a simple polygon. We do not make the bridges fully opaque; rather, they are ``semi-transparent'': we triangulate the simple polygon and do the staged illumination from $s$, but, as we go, each time a bridge is illuminated on one of its sides, at the next stage we consider it to be illuminated also on its other side, and continue the staged illumination. Let \opt denote an optimal path, and, abusing notation, also the number of links in it. In comparison to \opt, we are delayed by 1 link each time an edge of \opt crosses a bridge. That is, on every edge of \opt we have as many additional vertices as there are bridges that the edge crosses. Using a low-stabbing-number tree for the bridging \cite{manyppl}, we ensure that each edge of \opt crosses $O(\sqrt h)$ bridges, and thus we obtain an $O(\sqrt{h})$ multiplicative approximation.
\begin{figure}\centering
\includegraphics[page=1]{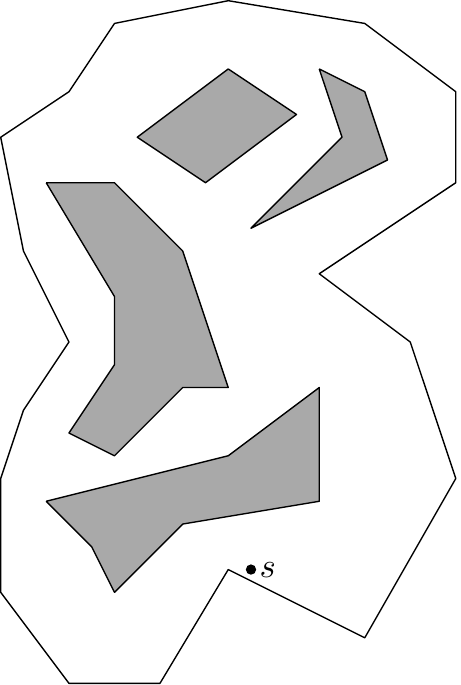}\hfil\includegraphics[page=2]{overview}\hfil\includegraphics[page=3]{overview}\vskip 10pt
\includegraphics[page=4]{overview}\hfil\includegraphics[page=5]{overview}\hfil\includegraphics[page=6]{overview}
\caption{Bridges are red. They block the light, but in the next stage they become windows and emit light on the other side.}\label{fig:overview}\end{figure}


As described above, the illumination takes $O(nh)$ time, which is quadratic in the worst case, since each of the $\Omega(n)$ triangles can potentially be discovered by light emanating from $h$ different bridges (Fig.~\ref{fig:nh}). To address this inefficiency, we declare a triangle opaque as soon as it is lit $m$ times, for a parameter $m\leq h$ (Fig.~\ref{fig:opaque}). We prove that with the right choice of $m$ this does not delay the illumination by too much, while decreasing the runtime of the illumination to $O(nm)$.
\begin{figure}\centering\includegraphics[page=1]{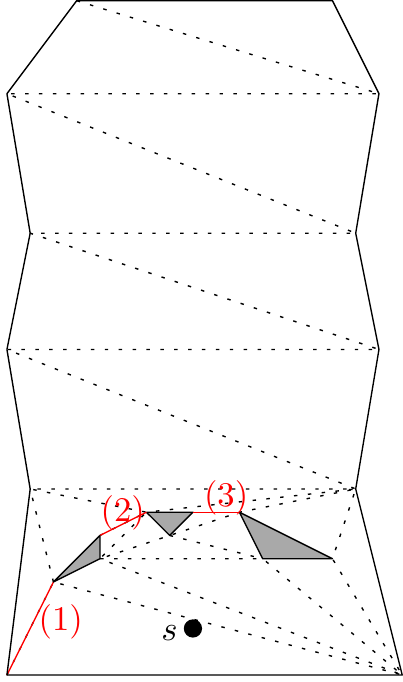}\hfil\includegraphics[page=2]{nh}\hfil\includegraphics[page=3]{nh}\\
(a)\hfil(b)\hfil(c)\hfil\vskip 20pt
\includegraphics[page=4]{nh}\hfil\includegraphics[page=5]{nh}\hfil\includegraphics[page=6]{nh}\\
(d)\hfil(e)\hfil(f)\hfil
\caption{(a)~The original domain, the bridges (shown in red) and a triangulation. (b)~The region illuminated at the first stage, i.e., the region that is seen from~$s$. $w$ is the only window. (c)~The region illuminated at the second stage from bridge~(1). (d)~Regions illuminated at the second stage from bridges~(1) and~(2). (e)~Regions illuminated at the second stage from bridges~(1), (2) and~(3). (f)~Regions illuminated at the second stage from all the bridges and the window~$w$. There are $\Omega(n)$ triangles of the triangulation that are illuminated from each of the bridges (and, in addition, from the window~$w$).}\label{fig:nh}\end{figure}
\begin{figure}\centering
\includegraphics[page=1]{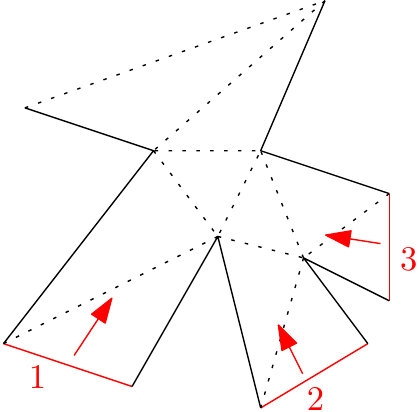}\hfil\includegraphics[page=2]{opaque}\hfil\includegraphics[page=3]{opaque}\vskip 10pt
\includegraphics[page=4]{opaque}\hfil\includegraphics[page=5]{opaque}
\caption{An example using $m=2$. Three windows are shown in red, labeled 1, 2, and 3; arrows indicate the direction in which the light will propagate from the windows. After windows~1 and~2 have propagated the light, some triangles (shown in orange) are 2-lit and are made opaque. Because of this, the light from window~3 lights up only the dark green triangle, but not anything further.}\label{fig:opaque}\end{figure}
\subsubsection{The algorithm}We compute a set $B$ of $h$ line segments, called \e{bridges}, such that the \e{cut polygon} $P^\c=P\setminus B$ is a (weakly) simple polygon (i.e., $B$ bridges the holes to $P$'s outer boundary) and such that any line intersects $O(\sqrt{h})$ bridges. Section~\ref{sec:bridging} details how this can be done in $O(n\log n+h\sqrt h\log h+n\sqrt{h})$ time. We triangulate the cut polygon $P^\c$, and do the staged illumination in it, modified as follows:
\newcounter{aa}
\begin{list}{\textbf{Modification M\arabic{aa}:}}{\usecounter{aa}}\item\label{mod1}We do not compute the windows exactly. Instead, we consider a triangle to be fully lit even if only part of it is illuminated. This simplifies the algorithm, since the boundary of the illuminated region now consists of edges of triangles. To account for a possible underestimation of the link distance that arises from treating partial illumination as full illumination, we add an additional link to our computed $s\-t$ path inside every lit triangle (Fig.~\ref{fig:partial_tri}). Overall this at most doubles the number of links in the output path in comparison to the stage at which $t$ is illuminated at the termination of our algorithm.
\item\label{mod2}Suppose a triangle having a bridge $b\in B$ as its side gets illuminated at stage $k$. As with the standard staged illumination, at stage $k$ we do not continue the illumination across~$b$ (because $b$ is an obstacle, a portion of the boundary of the cut polygon~$P^\c$). However, at stage $k+1$ we treat $b$ as one of the windows, and, thereby, continue the illumination on the other side of the bridge.
\newcounter{aaa} \setcounter{aaa}{\value{aa}}\end{list}
\begin{figure}\centering\hfil\includegraphics[page=1]{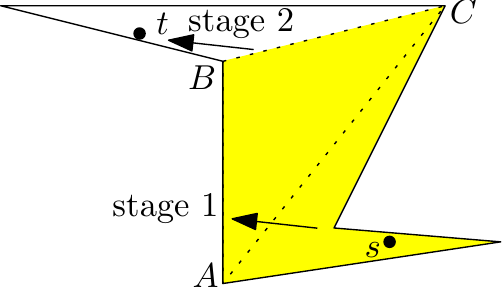}\hfil\includegraphics[page=2]{partial_tri}\hfil\caption{Left: Even though only part of the triangle $ABC$ is seen from~$s$, we fully light up the triangle at the first stage. Because of this, at the second stage we shine light from the side $BC$ and reach $t$ already at the second stage, even though there exists no 2-link $s\textrm-t$ path. Right: A 3-link path output by our algorithm. The middle link is added to account for the underestimation of the link distance brought by the modification {\bf M\ref{mod1}}.}\label{fig:partial_tri}\end{figure}

The second modification is the main difference between our illumination procedure and the standard one. The difference is best explained by identifying each element in $B\cup \{s\}$ with a ``color'' (there are thus $h+1$ colors). We will use the convention that anything colored with a color $b\in B\cup \{s\}$ gets prefix ``$b$-''; e.g., a window with color $b$ is a $b$-window, etc. 

Our illumination starts at $s$, and before any bridge is reached by it, all lit triangles are colored with $s$ (Fig.~\ref{fig:ill}). So far, the process does not differ from the standard illumination (modulo the modification {\bf M\ref{mod1}}). In particular, every triangle is $s$-lit only once because each $s$-window $w_s$ cuts out a unique dark portion $P_{w_s}\subset P^\c$ into which $w_s$ shines with the $s$-light; $P_{w_s}\cap P_{w_s'}=\emptyset$, $\forall w_s'\ne w_s$.
\begin{figure}
\begin{minipage}[c]{0.6\columnwidth}
\centering\includegraphics{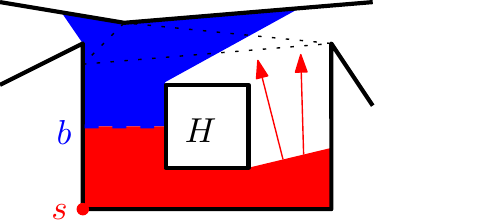}\caption{Point $s$ and the region illuminated at stage 1 are shown in red. Red arrows are some red illumination rays at stage 2. Bridge~$b$ (dashed blue) bridges a hole $H$ to the outer boundary. Crossing the bridge delays the illumination by one stage, and changes the color of the light to that of the bridge. Even though in the original domain $P$, the triangle $T$ (dotted) is seen directly from $s$, the triangle will be lit only at stage~2. $T$ is lit by both red and blue (after which the colors can be merged).}\label{fig:ill}
\end{minipage}\hfill
\begin{minipage}[c]{0.3\columnwidth}
%
\centering\includegraphics{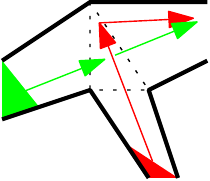}\caption{\apx has a green link, and green light is blocked by red at stage $k$ by a triangle $T$ (dotted): this means that $T$ has been reached with red after $k$ links (just as with green). Anything reachable with green via $T$ can be reached also with red, at the expense of a possible extra turn inside~$T$.}\label{fig:block}
\end{minipage}
\end{figure}

Assume now that a bridge $b$ is reached by the illumination at stage $k$. At stage $k+1$, $b$ starts shining $b$-light on its other side. Since the cut polygon~$P^\c$ is a simple polygon, each triangle is $b$-lit only once. However, nothing prevents a triangle from being seen both from an $s$-window and from a $b$-window (and, in general, from a window of any other color). Thus, any triangle can be discovered as many times as there are colors ($h+1$). This means that the overall time of the illumination is $O(nh)$.

Recall that \opt denotes an optimal path, and also the number of links in the path. Let \apx denote the $s\-t$ path found by our illumination procedure, and also the number of links in the path (i.e., the stage at which $t$ is reached).
\begin{lemma}\label{lem:nh}$\apx=O(\opt\,\sqrt h)$.\end{lemma}
\begin{proof}Because we use a low-stabbing-number tree for the bridging, any link of \opt intersects $O(\sqrt h)$ bridges. Thus, our illumination procedure will spend at most $O(\sqrt h)$ stages to propagate light along any one link of \opt.\end{proof}

Since $h$ can be $\Theta(n)$, we do not have a subquadratic algorithm yet. The $O(nh)$ running time is due to the possibility that a triangle can be lit by more than one color. To speed up the algorithm, we further modify the illumination procedure as follows:
\begin{list}{\textbf{Modification M\arabic{aa}:}}{\usecounter{aa}}\setcounter{aa}{\value{aaa}}\item\label{mod3}Fix a number $m=o(h)$. At any stage do the illumination color-by-color. Declare a triangle as an obstacle as soon as it has been reached by $m$ colors. That is, after being lit by $m$ colors, any triangle acts as an obstacle for light of any other color that tries to shine through it: it simply blocks the light.\setcounter{aaa}{\value{aa}}\end{list}	

Let \apxm denote the $s\-t$ path found by the staged illumination with modifications {\bf M\ref{mod1}-M\ref{mod3}}. With any triangle being discovered at most $m$ times, the running time to find \apxm goes down to $O(nm)$. To bound the number of links in the path (denote it by \apxm too), consider the path \apx. Its links are colored according to the illumination with modifications {\bf M\ref{mod1},M\ref{mod2}}. Suppose that one of the links is green. If green color is blocked due to the modification {\bf M\ref{mod3}} (say, at stage $k$ by a triangle $T$), we know that $T$ has been lit at stage $k$ by some other (at least $m$) colors. Any of these colors propagates into the region into which green would have propagated had it not been blocked (Fig.~\ref{fig:block}). Thus, we may have an extra turn in \apxm for each stage when a color of \apx gets blocked. Denoting by $X$ the overall number of stages at which blocking happens during execution of our algorithm, we have
\begin{lemma}\label{lem:apxm}$\apxm\le\apx+X$.\end{lemma}

We emphasize that the difference $\apxm-\apx$ depends on the number of \e{stages} at which blocking happens, not on the total number of blockings encountered throughout the algorithm---blockings that happen outside links of \apx do not contribute to the difference. Also, note that the path \apx and the colors of its links enter only the \e{analysis} for \apxm; in the algorithm for \apxm we, of course, do not compute \apx and do not (have to) know its links' colors.

What remains is to bound $X$. To do so we relate double-lighting to color-merging via the following observation: if a triangle $T$ can be lit both by a red color $r$ and a green color $g$ (i.e., if $T$ is seen from an $r$-window and a $g$-window) at stage $k$, then no triangle $T'\ne T$ may ever be $r$-lit through one side and be $g$-lit through another side after stage $k$. That is, if $T'$ is double-lit, the two colors $r,g$ must arrive through the same side of $T'$. Hence the two colors effectively can be merged into a single color. To see why this is the case, recall that any color lights up triangles following the dual of the triangulation of the cut polygon~$P^\c$ (which is a simple polygon). If red and green arrive at $T'$ via different sides, then there are two distinct $T\-T'$ paths (one followed by red light, one followed by green) in the dual---a contradiction (Fig.~\ref{fig:2lighting}).
\begin{remark}In terms of the dark portions $P_w\subset P^\c$ cut out by windows, the merging formally means that the collection of sets $P_w$, for $w$ ranging over all windows, is a laminar family (if two sets are not disjoint, then one is contained in the other). This is not as strong as a family of pairwise-disjoint sets (as is the case for staged illumination in simple polygons) but is good enough to ensure that each triangle effectively is discovered only once.\end{remark}
%
\begin{figure}\centering\includegraphics{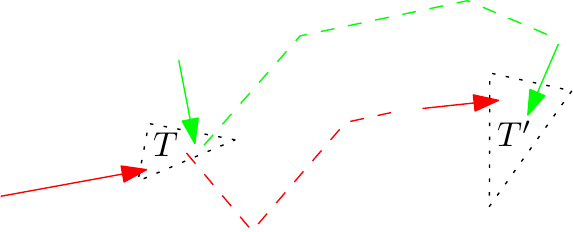}\caption{$T$ was $g$-lit and $r$-lit at some stage. If later $T'$ is $g$-lit and $r$-lit through different sides, then there are two distinct $T\textrm-T'$ paths in the dual of the triangulation of the simple polygon~$P^\c$.}\label{fig:2lighting}\end{figure}
\begin{lemma}\label{lem:X}$X\le h/m$.\end{lemma}
\begin{proof}Recall that we are considering blocking that occurs along a single path. There are $h+1$ colors to start with. Each time a blocking occurs, a triangle has been seen from windows of $m$ colors plus a window of the blocked color.\end{proof}
We emphasize that the blocking is algorithmic (i.e., we do block during the illumination) while the merging is used only in the analysis. (Of course, we could maintain the colors in a union-find data structure and merge them, but we do not need to do it: even if we propagate the colors unmodified, the merged colors will act as a single color exactly in the sense that no triangle will be discovered by more than one of them.)

From Lemmas~\ref{lem:nh}, \ref{lem:apxm}, \ref{lem:X} we have:
\begin{proposition}\label{prop:ill}If $B$ is given, then for any $m$, in $O(nm)$ time one can find a path with $O(\opt\,\sqrt h)+h/m$ links.\end{proposition}
Choosing $m=\Theta(\sqrt h)$, we obtain the main result of this section:
\begin{theorem}\label{thm:sqrth}An $O(\sqrt h)$-approximate minimum-link path can be found in $O(n\sqrt h+n\log n+h\sqrt h\log h)$ time.\end{theorem}
Our algorithm labels the triangles with approximate link distance to $s$ and thus also builds a linear-size approximate link distance map in subquadratic time (we recall that the exact map can have $\Theta(n^4)$ complexity \cite{so}).
\begin{remark}With the cutting of $P$ into $P^\c$ and propagating light on the other sides of bridges, our approach resembles, in some sense, working in the universal covering space of $P$ \cite{hs}.\end{remark}
\subsubsection*{$O(\sqrt h\log h)$-approximation with $m=1$}The bottleneck in the algorithm of Theorem~\ref{thm:sqrth} is computing the bridges $B$. If a more efficient solution is designed for the bridging, it will be of interest to do the illumination faster, too. There may also be other reasons to speed up the illumination: e.g., note that the bridges are good irrespective of the choice of $s$ and $t$; thus, one may use the same $B$ when finding minimum-link paths between different pairs $s,t$. (Of course, it would be more interesting to design a special data structure to do 2-point approximate link distance queries efficiently; this is however beyond the scope of the current paper.)

Our illumination is fastest with $m=1$ (i.e., when any triangle becomes an obstacle and blocks light as soon as it is lit by \e{any} color). However, with $m=1$ Proposition~\ref{prop:ill} does not give a multiplicative approximation guarantee (because \opt can be $o(\sqrt h)$). We now show that with yet another modification, one can use $m=1$ and guarantee an $O(\sqrt h\log h)$ approximation factor while only adding $O(h\log h)$ to the time of the illumination.

Specifically, our last modification is as follows:
\begin{list}{\textbf{Modification M\arabic{aa}:}}{\usecounter{aa}}\setcounter{aa}{\value{aaa}}\item\label{mod4}Identify colors with (not necessarily distinct) integers, and at each stage do the illumination color-by-color in decreasing order (with ties broken arbitrary). When a color first becomes active, it is identified with 1. When a color $g$ gets blocked by a color $r$, merge the colors into a single color $r+g$ (i.e., any color equals to the number of the colors merged into it).\end{list}

Note that with this modification, we do the color merging algorithmically. To do it efficiently, we store the colors, in the order in which we illuminate from them, in a binary search tree. The tree is updated each time a merge occurs. There can be at most $h$ merges, so the total time spent in updating the tree is $O(h\log h)$.

Let \apxl denote the number of links in the $s\-t$ path found by the staged illumination with modifications {\bf M\ref{mod1}-M\ref{mod4}}.
\begin{lemma}$\apxl=O(\apx\cdot\log h)$.\end{lemma}
\begin{proof}Consider a link $ab$ of the path \apx (the path found with modifications {\bf M\ref{mod1}-M\ref{mod2}}). Suppose that when we do the illumination with all modifications {\bf M\ref{mod1}-M\ref{mod4}}, the triangle containing $a$ is lit with a green color, identified with integer $g$. If green light fails to reach $b$, it must be blocked by some other color, say a red color, identified with integer $r$. This means that $r\ge g$. Since $r$ and $g$ bump into each other, they get merged into a single color $r+g\ge2g$. That is, whenever propagation is delayed at some stage, the color number at least doubles. Thus, after $O(\log h)$ such delays, the illuminating color is at least $h/2+1$, which means the we are illuminating with the highest possible color and no color remains to block light from reaching~$b$.\end{proof}

We thus obtain the following variation of Proposition~\ref{prop:ill}:
\begin{proposition}If $B$ is given, in $O(n + h\log h)$ time one can find a path with $O(\opt\,\sqrt h\log h)$ links.\end{proposition}
\begin{remark}The improvement brought by the modification {\bf M\ref{mod4}} may look surprising: the original illumination description did not specify in what order the illumination is done using different colors; that is, the order was arbitrary. With the modification, the order is still quite arbitrary because the tie breaking is arbitrary (and ties may happen often). The reason the modification helps is that the illumination is done in an arbitrary, but nevertheless systematic and stage-to-stage consistent way.\end{remark}

\subsection{Computing the bridges}\label{sec:bridging}We describe the algorithm for computing bridges, following the exposition from \cite[Section~7]{pedestrian}. The technique, called ``well-known'' already in \cite{pedestrian}, dates back to \cite{manyppl} and is based on using \e{low-stabbing-number} spanning trees for point sets. To our knowledge, such trees were first described in \cite{welzl}; see also \cite[Section~5]{matousekBook} and~\cite{aronov,chazelleWelzl,polyPart}.
\begin{lemma}[\cite{welzl}]\label{lem:stab}Let $R$ be a set of $h$ points in the plane. In $O(h\sqrt h\log h)$ time one can compute a (non-self-crossing) spanning tree $T$ of $R$ such that any line in the plane intersects $O(\sqrt h)$ edges of~$T$.\end{lemma}
Recall that the goal of the bridging is to compute a set $B$ of $h$ line segments (the bridges) to connect up the holes and the boundary of $P$ into a (weakly) simple polygon.

We surround $P$ with a large bounding box ${\cal B}$ and compute, in time $O(n\log n)$, a full triangulation within ${\cal B}$, of $P$ and its holes, as well as the portion of ${\cal B}$ outside the outer boundary of $P$ (we think of this region as another ``hole''). Refer to Fig.~\ref{fig:bridging}, top left.
\begin{figure}\centering
\hfil\includegraphics[page=1,scale=.5]{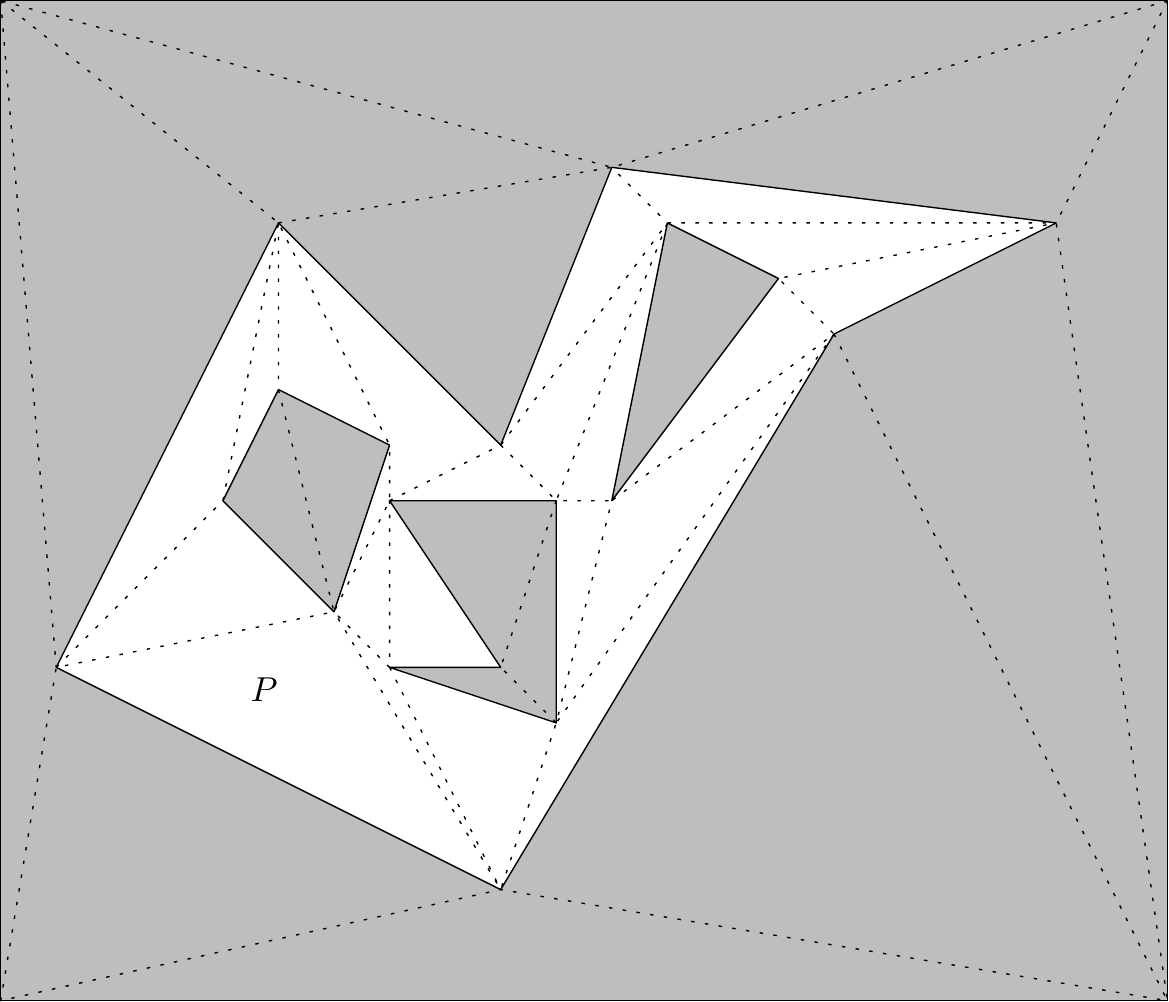}\hfil\includegraphics[page=2,scale=.5]{bridging}\hfil\\
\hfil\includegraphics[page=3,scale=.5]{bridging}\hfil\hfil\includegraphics[page=4,scale=.5]{bridging}\hfil
\caption{Top left: The triangulations (dotted). Top right: The spanning tree $T$ (thick edges) of the representative vertices (circles), and the intersections (squares). Bottom left: Each segment connects two different holes. Bottom right: The computed bridges.}\label{fig:bridging}\end{figure}

Next, pick one representative vertex from each hole ($h+1$ representatives altogether) and build the low-stabbing-number spanning tree $T$ of the representatives (the tree has $h$ edges). The edges of $T$ do not cross each other, but they may intersect boundaries of holes. To discover such intersection points, take each edge of $T$ and trace it using the triangulation. By Lemma~\ref{lem:stab}, each diagonal crosses $O(\sqrt{h})$ edges of $T$, so that the total number of triangles crossed by all edges of $T$ is $O(n\sqrt h)$. Thus, in total $O(n\sqrt h)$ time we can discover all of the intersection points. Refer to Fig.~\ref{fig:bridging}, top right. 

The intersection points subdivide the edges of $T$ into a total of $O(n\sqrt{h})$ smaller segments. Each segment either connects two boundary points of the same hole (the segment may lie fully inside the hole or fully outside it), or connects two different holes. Remove the former segments. Now each segment connects two different holes. Refer to Fig.~\ref{fig:bridging}, bottom left.

Examine the segments one by one, in an arbitrary order, and merge all holes into one (using a union-find data structure). That is, for each segment in turn, check whether it connects the boundary points of the same hole or of different holes. In the former case, remove the segment. In the latter case, keep the segment as one of the bridges in $B$ and merge the holes that it connects into one hole. Since we start with $h+1$ holes, $B$ will contain $h$ bridges when all holes are merged. Refer to Fig.~\ref{fig:bridging}, bottom right.

Finally, to obtain the (weakly) simple cut polygon $P^\c=P\setminus B$, slice $P$ along the bridges (each bridge becomes two edges of the cut polygon $P^\c$). It follows from the above that the cut polygon can be built in $O(n\log n+h\sqrt h\log h+n\sqrt h)$ time and that it has the desired property -- any line in the plane intersects any bridge $O(\sqrt h)$ times.

\section{\cori paths}\label{sec:cori}In this section $P$ is a \cori domain. To find a minimum-link \cori $s\-t$ path (and to construct a \cori link distance map from $s$) in $P$, we follow the general approach of \cite{aos}: Build $C$ trapezoidations of $P$, where the trapezoids in the trapezoidation $c\in C$ are obtained by extending maximal free $c$-oriented segments through vertices of $P$ (Fig.~\ref{fig:coriTraps}). Label the trapezoids with their link distance from~$s$. This way an $O(Cn)$-space structure is created that answers link distance queries in $O(C\log n)$ time. The labeling proceeds in $n$ \e{steps}, with label-$k$ trapezoids receiving their label at step $k$. Any label-$k$ trapezoid must be intersected by a label-$(k-1)$ trapezoid of a different orientation. Hence, step $k$ boils down to detecting all unlabeled trapezoids intersected by label-$(k-1)$ trapezoids. In terms of the \e{intersection graph} of trapezoids (which has nodes corresponding to trapezoids and edges corresponding to intersecting pairs of trapezoids), the labeling is done by Breadth First Search (BFS) starting from the $C$-oriented segments through~$s$.
\begin{figure}\centering\includegraphics{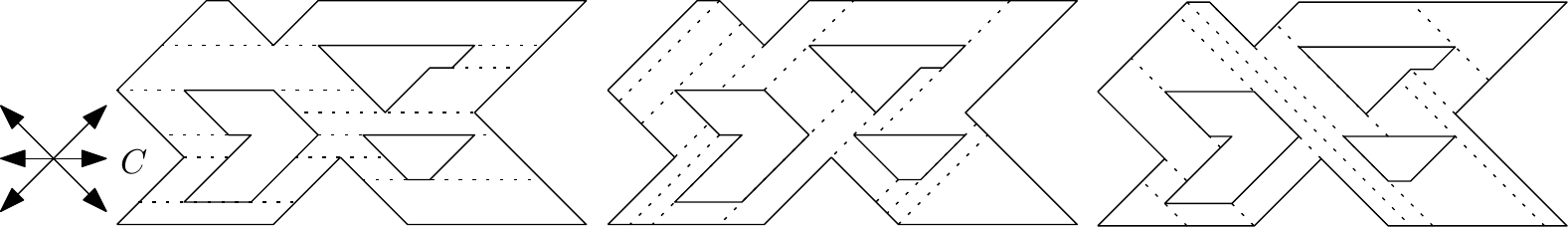}\\\vskip 10pt
\includegraphics[page=1]{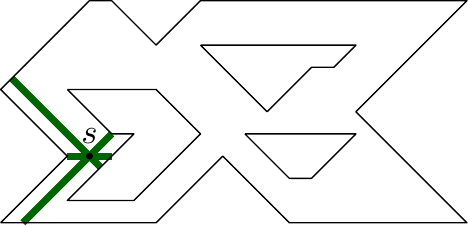}\hfil\includegraphics[page=2]{coriDist}\hfil\includegraphics[page=3]{coriDist}
\caption{Top: The trapezoidations. Bottom: The regions reachable with 1, 2 and 3 links.}\label{fig:coriTraps}\end{figure}

The idea of performing an efficient BFS in the intersection graph without explicitly building the (potentially quadratic-size) graph dates back to the work on minimum-link \e{rectilinear} paths \cite{surveylink,dn,lyw,ylw,ylwBends,ylwSIJCOMP,ohtsuki,sso,ia1,ia2,bergRect}. As noted already in \cite{ylw,dn}, Imai and Asano's \cite{ia1,ia2} data structure allows one to do the BFS---and hence to find a minimum-link rectilinear path---in $O(n\log n)$ time and space. Still, it was not until the 1991 work of Das and Narasimhan \cite{dn} (and the lesser known work of Sato, Sakanaka and Ohtsuki \cite{sso}) that an optimal, $O(n\log n)$-time $O(n)$-space BFS implementation was developed (see also Appendix~\ref{app:Lee}). In the implementation, one step of the BFS is reduced to a pair of sweeps---the UpSweep and the DownSweep.

The data structures employed in \cite{dn,sso} are simpler than the (much more general-purpose) structure of Imai and Asano \cite{ia1,ia2}, but they are still more complicated than one would hope for the basic problem of finding rectilinear paths amidst rectilinear obstacles.\footnote{The conference paper \cite{dn} admits omitting the details in several places; to our knowledge, no fuller version exists. In particular, we were not quite able to parse the description in \cite{dn}, in which the light propagates from ``portions'' of ``portions'' of ``outermost'' segments discovered at the previous step: according to the 5th paragraph on p.~265, at any step light is directed from ``portions'' of the fronts, but the fronts themselves seem to be ``portions'' of the ``outermost'' segments lit at the previous stage (4th paragraph).} In Section~\ref{sec:rect} we present a simplified implementation of the BFS step in the intersection graph. Our modification does not affect the asymptotic time and space optimality of the algorithm. The crux of the simplification is the use of a single tree for storing the intersection of the sweepline with the trapezoids that were lit at the previous step.

Another minor modification in our algorithm comes from performing only one (upward) sweep at any step of the BFS. The sweep starts from what we call the ``pot'' trapezoids---those into which the different-orientation trapezoids lit at the previous step are ``planted''. The planting is nothing but a means to initialize the sweep (without the planting, it is not clear how to efficiently discover even a single edge in the intersection graph). While for the rectilinear case our modification saves only a factor of 2 in running time, it serves as the basis of our improvements for the general case of $C$-oriented paths with $C>2$ (Section~\ref{subsec:cori}).
\begin{remark}In Section~\ref{sec:coriApx} we also use the original algorithm of \cite{dn}---with both the UpSweep and DownSweep---to give a more time-space efficient 2-approximation algorithm for minimum-link \cori paths.\end{remark}

\subsection{Rectilinear paths in rectilinear domains}\label{sec:rect}In this subsection $P$ is a rectilinear domain, and we want to build the rectilinear link distance map from $s$. For simplicity of exposition, we make a non-degeneracy assumption that no two edges of $P$ are supported by the same line; it is straightforward to handle the degenerate cases. As in \cite{dn}, we start by forming 2 decompositions of $P$: one by extending maximal free horizontal segments supported by horizontal edges of $P$, and the other by extending vertical segments. The horizontal decomposition is stored in a linked structure that links any cell to its upper neighbors. The vertical decomposition is stored analogously. Both decompositions are rectangular; however, to connect better to our techniques for general $C$-oriented paths (Section~\ref{subsec:cori}), we call the rectangles in the horizontal decomposition \e{trapezoids}, and those in the vertical decomposition we call \e{parallelograms}. The horizontal edges of a trapezoid are its \e{bases}, and the vertical edges are its \e{sides}. For a parallelogram, on the contrary, its vertical edges are bases, and horizontal edges are the sides. Any side is thus a subset of an edge of~$P$.

Our goal is to label the cells of the decompositions (i.e., trapezoids and parallelograms) with link distance from $s$. If a trapezoid (resp., parallelogram) has label $k$ then every point inside the trapezoid (resp., parallelogram) can be reached with a $k$-link path whose last link is horizontal (resp., vertical), and some point inside the trapezoid cannot be reached with fewer than $k$ links. As discussed above, the labeling amounts to BFS in the trapezoids--parallelograms intersection graph, starting from maximal free horizontal and vertical segments through $s$ (which are labeled with 1). Step $k$ of the BFS is as follows: given a subset $S_*^{k-1}$ of parallelograms (those labeled $k-1$), find all unlabeled trapezoids intersected by parallelograms in $S_*^{k-1}$. Denote the set of the sought trapezoids by $S^k$. (Strictly speaking, this is only half of the BFS step; the other step is analogous, with the roles of trapezoids and parallelograms reversed.) The idea is to discover the trapezoids from $S^k$, in the order of increasing $y$-coordinate of the lower bases, by sweeping a horizontal line upwards.

The remaining question is how to initialize and maintain the sweep event queue. We do not want to insert \e{all} unlabeled trapezoids into the queue (as it could lead to linear-size queues on a linear number of the BFS steps). On the other hand, the trapezoids from $S^k$ \e{must} appear in the step-$k$ queue at some point. Of course, ideally, \e{only} the trapezoids from $S^k$ should ever appear in the queue; we, however, did not see how to enforce this.

Instead, we insert a \e{superset} of $S^k$ into the queue. Every trapezoid $T$ inserted in the queue at step $k$ is either itself intersected by a parallelogram from $S_*^{k-1}$, or it has a lower neighbor intersected by a parallelogram from $S_*^{k-1}$ (or both). In the former case, the label of $T$ may not be smaller than $k-2$ and may not be larger than $k$. In the latter case, the label of $T$ is not larger than $k+2$ and not smaller than $k-4$ (Fig.~\ref{fig:rect}, left). That is, the label of any trapezoid queued at step $k$ belongs to the interval $[k-4,k+2]$. Viewed differently, this means that any trapezoid is processed during at most 7 steps (7 is not a tight bound, but all we need is that it is a constant). This gives an $O(n)$ bound on the total number of events over all $n$ BFS steps. Since there are $O(n)$ trapezoids, the claimed $O(n\log n)$ time and $O(n)$ space bounds for the algorithm follow once we show how to process an event in $O(\log n)$ time.

The correctness of the sweep follows from the standard invariant: all trapezoids from $S_k$ whose lower bases are below the sweepline are labeled. We now give the details of the sweep: its initialization, sweepline status and the events.
\begin{figure}\centering\includegraphics[scale=.9]{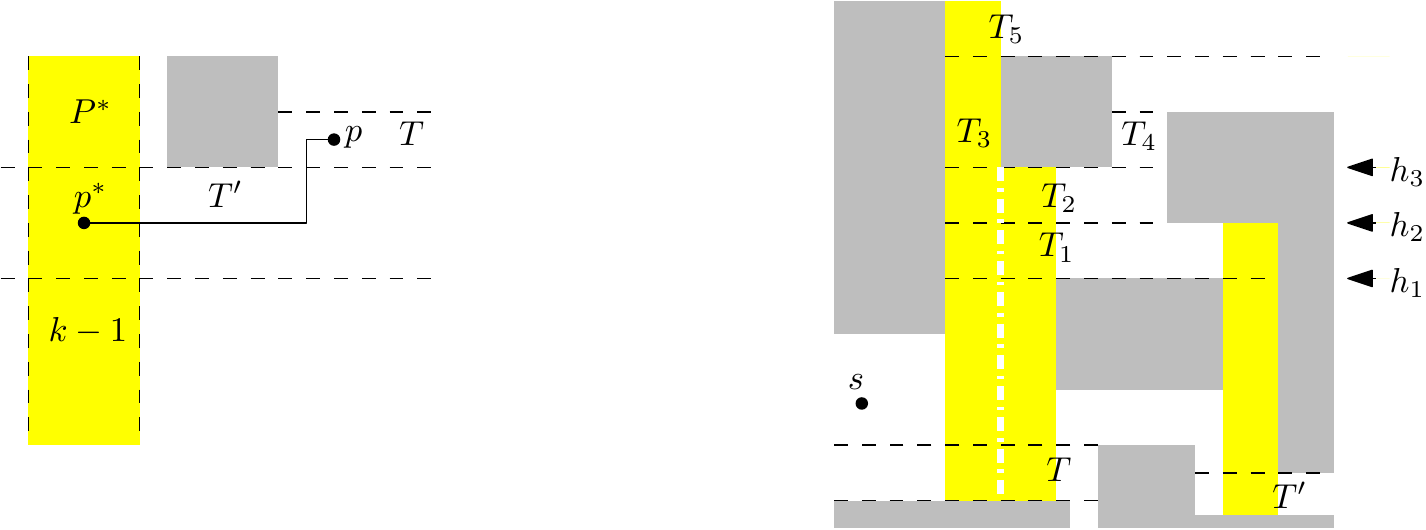}\caption{Left: If $T$ has a lower neighbor $T'$ intersected by a parallelogram $P^*\in S_*^{k-1}$, then from any point $p\in T$ there exists a 3-link path to some point $p^*\in P^*$. Thus, the label of $T$ is within 3 of $k-1$, the label of $P^*$. Right: Three parallelograms from $S_*^{2}$ are shown yellow. The two left parallelograms are planted into $T$, the rightmost parallelogram is planted into $T'$. Some of the events are as follows: At $h_1$, the trapezoid $T_1$ is the event. Since its lower base overlaps with intervals in the sweepline status, $T_2$ is inserted into the event queue. Since $T_1$ is dark, it gets its label~3. At $h_2$, the first event is the update of the sweepline status. The next event at $h_2$ is $T_2$: $T_3,T_4$ are inserted into the queue and $T_2$ is labeled. At $h_3$, first the sweepline status is updated; then, $T_3$ is processed: since it is intersected by the parallelograms from $S_*^{2}$ (as witnessed by the sweepline status intervals), $T_5$ enters the queue and $T_3$ is labeled. Finally, $T_4$ is processed: its lower base does not overlap intervals from the status, so it remains unlabeled.}\label{fig:rect}\end{figure}
\paragraph{Planting}We say that a parallelogram $P^*\in S_*^{k-1}$ is \e{planted} into a trapezoid $T$ if the lower side of $P^*$ overlaps with the lower base of $T$ (Fig.~\ref{fig:rect}, right). We say also that $T$ is the \e{pot} of $P^*$. We initialize the event queue by inserting all pots into it. The crucial observation is that each parallelogram is planted into exactly one pot (even though a pot can have many parallelograms planted side-by-side into it). Moreover, it is straightforward to augment the data structures storing the decompositions so that for any parallelogram its pot can be determined in $O(1)$ time (store with each parallelogram the horizontal edge of $P$ that supports its lower side, and store with each edge of $P$ the trapezoid whose lower base supports the edge).
\paragraph{Parallelogram events}The sweepline status is the intersection of the sweepline with the (interiors of) parallelograms from $S_*^{k-1}$. The status thus is a set of disjoint (open) intervals. The status changes at \e{parallelogram events} when the sweepline reaches sides of parallelograms. If the event is an upper side, the side is removed from the status. If the event is a lower side, it is added to the status (the parts of intervals already in the status overlapping with the newly added interval are removed). The total number of parallelogram events over all BFS steps is $O(n)$ since each label-$(k-1)$ parallelogram contributes 2 events at step $k$. Because the intervals in the status are disjoint, we can keep them in any ordered structure, e.g., a balanced binary search tree indexed by left endpoints of the intervals (the more complicated data structures in \cite{dn} were utilized to maintain ``fronts'' composed of merging and splitting ``windows''). Clearly, the tree handles any of the following three operations in amortized $O(\log n)$ time: (1)~adding an interval (we do not merge the intervals into superintervals because for backtracking we need to know from which parallelogram a trapezoid was lit), (2)~removing part of an interval that hits an obstacle edge, and (3)~checking whether any of the intervals overlaps with a given query interval. In the last operation, the query interval is a trapezoid lower base. We need it for the trapezoid events, described next.
\paragraph{Trapezoid events}The main events in the sweep are \e{trapezoid events} that occur when the sweepline reaches a lower base of a trapezoid (some of the trapezoid events happen simultaneously with parallelogram events; in this case parallelogram events take priority). Suppose that a trapezoid $T$ is the event. We check whether the lower base of $T$ is intersected by the intervals in the sweepline status. If yes, we insert all (at most 2, due to non-degeneracy) upper neighbors of $T$ into the event queue (recall that to facilitate the insertion, the upper neighbors are linked from $T$). In addition, if $T$ is unlabeled, we label it with $k$. Refer to Fig.~\ref{fig:rect}.

\subsection{$C$-oriented paths in \cori domains}\label{subsec:cori}Let $P$ be a $C$-oriented domain. Our goal is to build the $C$-oriented link distance map from $s$. As noted in \cite{aos}, the efficient methods developed to perform BFS in the trapezoids intersection graph for the rectilinear version do not extend to $C$-oriented paths when $C>2$. Let us look closely at why this is the case. One reason is that for $C>2$ some trapezoids may get labeled only partially during a BFS step (Fig.~\ref{fig:split}). This complicates the BFS because the intersection graph changes from step to step, and, in the final link distance map, trapezoids may get split into subtrapezoids. The partial labeling and splitting are due to the possibility that two different-orientation trapezoids do not ``straddle'' each other; instead they both may be ``flush'' with an obstacle edge whose orientation is different from the orientations of both trapezoids (this was not the case in the rectilinear version, since there were only 2 orientations). However, such a flush intersection can be read off easily from lists of incident trapezoids stored with each edge of~$P$. Thus, discovering partially labeled trapezoids becomes the easy part of the algorithm. For the remaining part we do a sweep for each pair of orientations in~$C$.\begin{figure}\centering\includegraphics[width=.3\columnwidth,page=2]{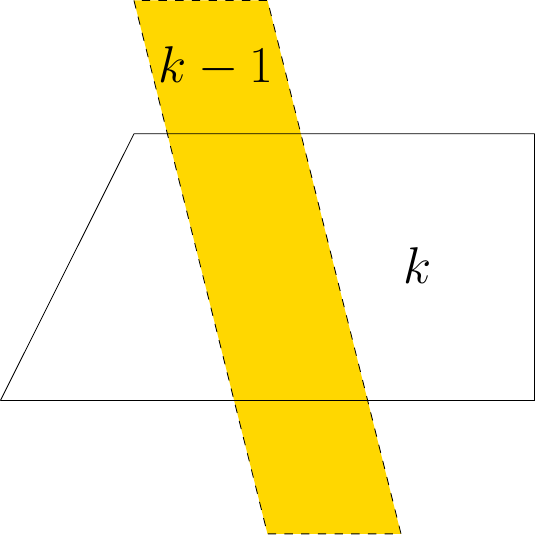}\hfil\includegraphics[width=.3\columnwidth,page=4]{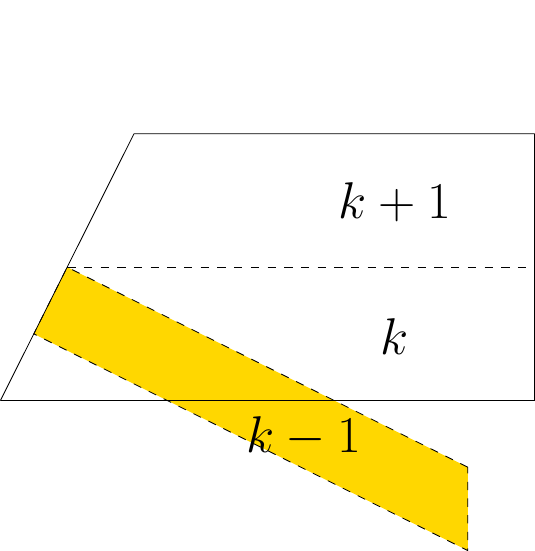}\caption{Left: A trapezoid gets fully labeled. Right: A trapezoid is labeled only partially at step~$k$.}\label{fig:split}\end{figure}

Specifically, after the partially labeled trapezoids are processed, we are left with discovering unlabeled trapezoids ``fully straddled'' by trapezoids labeled at the previous step. As with the rectilinear case, it is the straddling that leads to a superlinear-size intersection graph and makes a subquadratic algorithm less trivial. It is tempting to reuse here the sweeping techniques developed for the rectilinear version. The stumbling block, though, is choosing the direction of the sweep. Indeed, no matter in which direction the sweep proceeds, the intersection of the sweepline with a trapezoid does not change only at discrete ``events'': while the sweepline intersects a non-parallel side of the trapezoid, the intersection changes continuously. The good news is that this is the \e{only} reason for a continuous change of the intersection. This prompts us to get rid of the non-parallel sides of the trapezoids by clipping them into parallelograms, with the new sides parallel to the sweepline. After the clipping (and planting the parallelograms appropriately) is done, we are able to reuse the rectilinear-case machinery and finish the BFS step with $C(C-1)$ sweeps---one per (ordered) pair of orientations.
Overall we obtain an $O(C^2n\log n)$-time $O(Cn)$-space algorithm.

The main difference between our algorithm and that of \cite{aos} is the separate treatment of flush and straddling trapezoids. This allows us to use only elementary data structures and improve the time-space bounds to $O(C^2n\log n)$ and $O(Cn)$. We now describe the details.

\subsubsection*{Definitions and preliminary observations}We start by defining the notation and recalling some results from \cite{aos}. 
Any trapezoid $T$ has two opposite edges belonging to the boundary of $P$; these edges are \e{sides} of $T$. The other two edges are $T$'s \e{bases}; the bases are parallel segments whose orientation belongs to $C$. For an orientation $c\in C$, a \e{$c$-segment} is a segment with orientation~$c$. A \e{$c$-trapezoid} $T$ has $c$-segments as bases. A \e{subtrapezoid} of $T$ is a $c$-trapezoid cut out from $T$ by one or two $c$-segments.
\paragraph{$c$-distance}A \e{$c$-path} is a path (starting from $s$) whose last link is a $c$-segment. A point $p\in P$ is at \e{$c$-distance} $k$ from $s$ if $p$ can be reached by a $k$-link $c$-path (but not by a $(k-1)$-link $c$-path).
\paragraph{The output: $c$-maps}The $c$-distance equivalence decomposition of $P$ (the \e{$c$-map}) is the partition of $P$ into $c$-trapezoids such that the $c$-distance to any point within a cell is the same. If the $c$-distance to points in a $c$-trapezoid of the $c$-map is $k$, then the $c$-trapezoid has \e{label} $k$. Using the illumination analogy we also say that the trapezoid is \e{lit} at step $k$. Unlit trapezoids are \e{dark}. Denote the set of $c$-trapezoids lit at step $k$ by $S_c^k$.
\paragraph{Convention: $c$ is horizontal and implicit}In our algorithm, finding $S_c^k$ is completely identical to (and independent from) finding $S_{c^*}^k$ for any $c^*\in C\setminus c$. In what follows we focus on finding $S_c^k$. Where it creates no confusion, we omit the subscript $c$ and the prefix ``$c$-''; e.g., ``path to trapezoid in $S^{k-1}$'' means ``$c$-path to $c$-trapezoid in $S_c^{k-1}$'', etc. We let $C^*$ denote $C\setminus c$, and use $c^*$ for a generic orientation from $C^*$. Assume, without loss of generality, that $c$ is horizontal.
\paragraph{$c$-map via trapezoidation refinement}Denote by $D^k$ the ``at-most-$k$-links map'', i.e., the trapezoidation whose trapezoids are of $k+1$ types---dark trapezoids, and trapezoids lit at steps $1,\dots, k$. All trapezoids in $D^k$ are maximal --- the dark trapezoids are maximal in the sense that each is a maximal trapezoid every point of which has distance $k+1$ or larger, and the lit trapezoids are the same as they are in the final map (i.e., trapezoids lit at step $k$ are exactly $S^k$). By definition, lit trapezoids from $D^{k-1}$ remain the same in $D^k$; in particular, $D^n$ is the $c$-map. Let $T'$ be a dark trapezoid from $D^{k-1}$. The crucial (albeit obvious) observation about minimum-link paths is (cf.\ \cite{suri,dn,aos}):
\begin{obs}\label{obs:main}There exists a $k$-link path to a point $p\in T'$ if and only if there exists a $(k-1)$-link $c^*$-path $\pi^*$ to some point $q\in T'$ that has the same $y$-coordinate as $p$.\end{obs}
If $\pi^*$ enters $T'$ through its lower base, then any point whose $y$-coordinate is smaller than the $y$-coordinate of $p$ is also reachable by a $k$-link path. Thus, the set of points of $T'$ reachable by $k$-link paths entering $T'$ through its lower (resp., upper) base is a subtrapezoid, $T'_l$ (resp., $T'_u$), of $T'$ (cf.\ \cite[Fig.~5]{aos}). If $T'_l\cup T'_u=T'$, we say that $T'$ is \e{fully} lit; otherwise, it is \e{partially} lit. Note that in the latter case the subtrapezoid $T'\setminus(T'_l\cup T'_u)$ will be fully lit in $D^{k+1}$, since orientations of edges of $P$ belong to $C$. Thus, $D^k$ is a refinement of $D^{k-1}$, and overall we have:
\begin{obs}\label{obs:CoriDomain}No trapezoid from $D^0$ is split into more than 3 trapezoids in $D^n$.\end{obs}
The trapezoidation $D^0$ is just the trapezoidal decomposition of $P$ with maximal free segments supported by vertices of $P$. Because $D^0$ is linear-size, by Observation~\ref{obs:CoriDomain} so is $D^n$, the $c$-map. To answer a link distance query from a point $q\in P$ to $s$, it is sufficient to locate $q$ in each of the $c$-maps and choose one in which the trapezoid of $p$ has the smallest label.

The above discussion reestablishes the space and query time results of \cite{aos}: a size-$O(Cn)$ structure can be built to answer link distance queries in $O(C\log n)$ time. Two approaches were presented in \cite{aos} to construct the $c$-maps: one using $O(C^2n\log n)$ time and space, the other using $O(C^2n\log^2 n)$ time and $O(C^2n)$ space. We construct the $c$-maps in $O(C^2n\log n)$ time and $O(Cn)$ space. Our general approach is still the same as in \cite{aos}: at step $k=1,\dots, n$, starting from $D^{k-1}$, build $D^k$ by identifying $S^k$---the trapezoids lit at step $k$. We find the lit trapezoids differently from \cite{aos}, without employing complicated data structures. The details follow below.
\subsection*{Intersection types}\label{subsec:int}Restated in our terms, Observation~\ref{obs:main} means that a dark trapezoid $T'\in D_c^{k-1}$ gets fully or partially lit at step $k$ if and only if it is intersected by some (different-orientation) trapezoid $T^*\in S_{c^*}^{k-1}$ lit at step $k-1$. We distinguish between two types of trapezoids intersection (Fig.~\ref{fig:types}).
\begin{definition}\label{def:flush}$T',T^*$ are \e{flush} if a side of $T'$ overlaps with a side of $T^*$. We say that $T'$ is (fully or partially) \e{flush-lit} by $T^*$.\end{definition}
A flush trapezoid $T'$ is lit fully or only partially depending on whether its side fully or only partially overlaps with the sides of the trapezoids in $\bigcup_{c^*}S_{c^*}^{k-1}$ that are flush with $T'$.
\begin{definition}\label{def:straddle}$T',T^*$ \e{straddle} each other if both bases of $T'$ intersect both bases of $T^*$. We say that $T'$ is (fully) \e{straddle-lit} by $T^*$.\end{definition}
In particular, if a side of $T^*$ overlaps with a base of $T'$ or vice versa, then $T',T^*$ are counted as straddling, not as flush.
\begin{figure}\centering\includegraphics{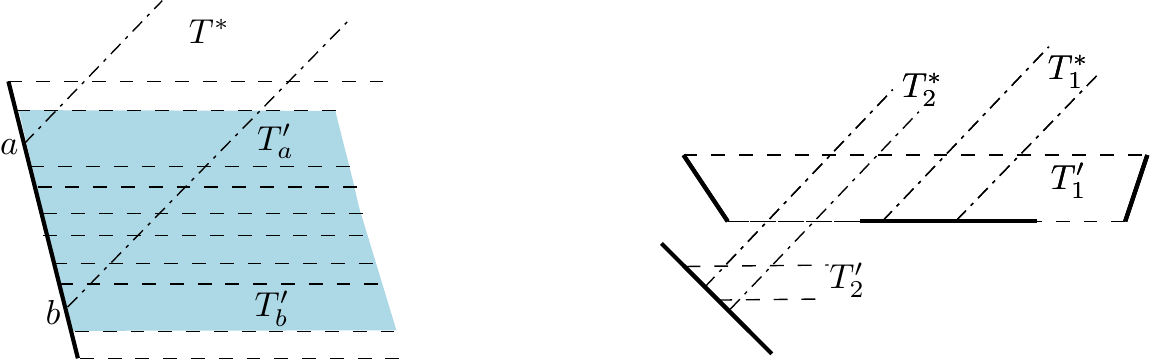}\caption{Intersection types. Left: $T^*$ is flush with the shaded trapezoids. $T'_b$ will be split in $D^k$ unless more of it is lit by another trapezoid. $T_a$ will be the pot for the parallelogram cut out of $T^*$ by the $c$-segment through~$a$. Right: $T^*_1,T^*_2$ straddle $T'_1$. The parallelogram cut out of $T^*_1$ is planted into $T'_1$. The pot for the parallelogram cut out of $T^*_2$ is a trapezoid $T'_2$ flush with~$T^*_2$.}\label{fig:types}
\end{figure}

Note that flush intersection and straddling are the only possible ways for two trapezoids from $D_c^{k-1},D_{c^*}^{k-1}$ to intersect: because vertices and sides of trapezoids belong to edges of $P$, if one trapezoid has a vertex inside the other, then the vertex is on an edge of~$P$ and the trapezoids are flush. (The flush intersection may be degenerate when the trapezoids just share a common vertex.) Also note that a straddle-lit trapezoid is necessarily fully lit.

With Definitions~\ref{def:flush} and~\ref{def:straddle}, step $k$ of the algorithm can be completed as follows: Find dark $c$-trapezoids flush with trapezoids from $S_{c^*}^{k-1}$, and (fully or partially) light them: check for each found flush trapezoid $T'$ whether it is flush-lit fully or partially. For each partially flush-lit trapezoid determine the dark subtrapezoid $T''\subset T'$, and label with $k$ the (sub)trapezoid(s) that form $T'\setminus T''$. After this has been done for all $c^*\in C^*$, i.e., after all flush trapezoids are processed, any dark trapezoid will either fully remain dark in $D^k$ or will be fully straddle-lit (i.e., there will be no more partial lighting and splitting). So what remains is to straddle-light $c$-trapezoids. For that we clip each $c^*$-trapezoid $T^*\in S_{c^*}^{k-1}$ to a $(c,c^*)$-parallelogram $P^*$, using $c$-segments going through vertices of $T^*$. We plant $P^*$ into the $c$-trapezoid $T'$ to which the lower base of the parallelogram belongs and do a sweep analogous to the rectilinear case to discover the trapezoids intersected by the parallelograms.

The clipping-planting-sweeping is repeated for each $c^*\in C^*$, i.e., overall, to straddle-light the $c$-trapezoids in $S_c^k$, we perform $C-1$ sweeps, one per $c^*\in C^*$. We proceed to a detailed description of the flush- and straddle-lightings.
\subsection*{Flush-lighting}\label{app:flush}Sides of trapezoids belong to edges of $P$. We say that an edge $e$ \e{supports} a trapezoid if one of its sides belongs to $e$. We maintain the ordered list $L_c(e)$ of $c$-trapezoids supported by $e$. Similarly, we store with each trapezoid the 2 edges of $P$ that support its sides. The flush-lighting is then done as follows: For every $c^*$-trapezoid $T^*\in S_{c^*}^{k-1}$ and each edge $e$ that supports $T^*$, locate the vertices $a,b$ of $T^*$ (lying on $e$) in the list $L_c(e)$. All (dark) trapezoids lying between $a$ and $b$ are labeled $k$. One of the trapezoids $T'_a,T'_b$ containing $a,b$ in the interior of the side is marked to be split, at the end of flush-lighting, by a horizontal cut through $a$ or $b$, unless more of the trapezoid is flush-lit by another trapezoid. Refer to Fig.~\ref{fig:types}, left.
\subsection*{Straddle-lighting}\label{app:straddle}Now that all flush-lit trapezoids have been found and lit, we have to find straddle-lit trapezoids. Fix $c^*\in C^*$.

Clip each $c^*$-trapezoid $T^*\in S_{c^*}^{k-1}$ to the parallelogram $P^*$ using horizontal lines through vertices of $T^*$. Denote the set of the obtained parallelograms by $S^\c$. Any $c$-trapezoid straddled by $T^*$ is also straddled by $P^*$, and thus straddle-lighting with $c^*$-trapezoids is equivalent to finding dark trapezoids intersected by parallelograms from $S^\c$. This can be accomplished with a sweep, called the \e{$(c,c^*)$-sweep}, similarly to the case $C=2$ from Section~\ref{sec:rect}. Below we describe the few differences.

Some $c^*$-trapezoids from $S_{c^*}^{k-1}$ (such as, e.g., trapezoid $T_1^*$ from Fig.~\ref{fig:types}, right) have lower bases supported by $c$-edges of $P$. Planting the parallelograms cut out from such trapezoids is identical to the rectilinear case---the pots are read off directly from the trapezoidations $D_c,D_{c^*}$ augmented with the little auxiliary information as described in Section~\ref{sec:rect}. The rest of the trapezoids from $S_{c^*}^{k-1}$ (such as, e.g., trapezoid $T_2^*$ from Fig.~\ref{fig:types}, right, or $T^*$ from Fig.~\ref{fig:types}, left) are flush with trapezoids from $D_c^{k-1}$. The pot $T'$ for the parallelogram $P^*$ cut out from such a trapezoid $T^*$ can be determined from the list $L_c(e)$, where $e$ is the edge supporting $T^*$ and $T'$: all that is needed is to locate in which trapezoid from the list the vertex of $T^*$ lands (and if the vertex of $T^*$ is a common vertex of two $c$-trapezoids, then the upper of the two is chosen as the pot). Note that if $T^*,T'$ are flush, then $P^*$ may not be planted into $T'$ ``all the way to the bottom'', and so the parallelogram event corresponding to the lower base of $P^*$ may not coincide with any trapezoid event; this means that the sweepline does not intersect $P^*$ at the time when $T'$ is the event, and hence the pot is not straddle-lit. This is not a problem because the pot is flush-lit by $T^*$ anyway, and no other $c$-trapezoid is intersected by $P^*$ before the parallelogram comes out of the pot (by which time the intersection of the sweepline with $P^*$ is already present in the sweepline status). Thus, the only fix to the sweep that we need is unconditional insertion of the (upper) neighbors of $T'$ into the event queue, even when the lower base of $T'$ does not intersect the sweepline-status intervals.

The final, minor difference from the rectilinear case is that the parallelograms in $S^\c$ are not vertical. To account for this, every operation on the sweepline status is preceded by a horizontal shift of the processed interval. Specifically, recall that the operations supported by the interval tree are insertion, deletion, and query of an interval. If the operation is performed when the sweepline is at height $h$, we shift the interval by $h/\tan\alpha$ before the operation, where $\tan\alpha$ is the slope of $c^*$.

We emphasize that clipping by the $c$-segments is done only to find $c$-trapezoids straddle-lit by $c^*$-trapezoids. After the $(c,c^*)$-sweep completes, the $c^*$-trapezoids are ``unclipped'' back to what they were (and in general, during a $(c_1,c_2)$-sweep, $c_2$-trapezoids lit at the previous step are only temporarily clipped into $(c_1,c_2)$-parallelograms using $c_1$-segments through the vertices).
\subsection*{Analysis}\label{cori:anal}Any of the $O(Cn)$ trapezoids in the final $c$-maps is either flush-lit or straddle-lit. Flush-lighting takes overall $O(C^2n\log n)$ time. Indeed, for every trapezoid $T^*$ that flush-lights $c$-trapezoids through an edge $e$, it takes $O(\log n)$ time to locate the vertices $a,b$ of $T^*$ (lying on $e$) in the list $L_c(e)$. Overall there are $O(Cn)$ trapezoids $T^*$, and for each we have to locate the vertices $a,b$ in the $C-1$ lists $L_c(e)$. Thus the locating takes overall $O(C^2n\log n)$ time. After the vertices $a,b$ have been located, it takes $O(n_e)$ time to label each (dark) trapezoid $T'$ supported by $e$ (here $n_e$ is the number of the trapezoids that are flush with $T^*$). Again, overall there are $O(Cn)$ trapezoids $T'$, and each can be flush-lit from at most $C-1$ directions. Thus the total time spent in the labeling (not counting the time spent in locating the vertices $a,b$) is $O(C^2n)$.

As for straddle-lighting, by the non-degeneracy assumption, any trapezoid has $O(1)$ neighbors. Thus, processing an event during any of the sweeps involves a constant number of priority queue and/or interval tree operations, i.e., $O(\log n)$ time per event. To bound the number of events, observe that exactly as in the rectilinear case, any trapezoid inserted in the event queue at step $k$ is either itself intersected by a parallelogram from $S^\c$, or has a lower neighbor intersected by a parallelogram from $S^\c$. Thus, just as in the rectilinear case, any trapezoid enters the event queue on at most 7 consecutive BFS steps. At any step $k$, a $c$-trapezoid may appear in the event queue during each of the $C-1$ $(c,c^*)$-sweeps. Thus, since there are $O(Cn)$ trapezoids, we have $O(C^2n)$ events, and the total running time of straddle-lighting is $O(C^2n\log n)$.

As for the space, the interval tree uses $O(n)$ space at a single sweep; moreover, because we do the sweeps in different directions independently, we never need more than a single interval tree during the execution of the algorithm. The dominating term is thus the storage of the $C$ trapezoidations, each requiring $O(n)$-space.
\begin{theorem}\label{thm:cori}An $O(Cn)$-size data structure can be built to answer $C$-oriented link distance queries in $C$-oriented domains in $O(C\log n)$ time; a minimum-link path can be output in additional time proportional to the distance. The preprocessing time and space are $O(C^2n\log n)$ and $O(Cn)$.\end{theorem}

\subsection{Extensions}We apply our methods to compute minimum-link $C$-oriented paths in arbitrarily oriented domains, to give a faster 2-approximation algorithm for finding minimum-link \cori paths, and to quantify what unrestricted-orientation paths are approximable by $C$-oriented ones.
\subsubsection{$C$-oriented paths in arbitrary domains}\label{sec:coriinarb}
    \ifcgta
\else

A simple generalization allows us to compute the $C$-oriented link distance map also in a domain that is not necessarily $C$-oriented. Examining the algorithm from the previous section, we see that orientations of the \e{sides} of the trapezoids did not play any role in the algorithm (what was important is that the \e{bases} of the trapezoids are $C$-oriented). The only place in which $P$'s $C$-orientedness was used is the algorithm's analysis: Observation~\ref{obs:CoriDomain} and its corollary that the final $c$-map is linear-size. That is, even if $P$ is an arbitrary domain, we can define the $c$-map as the decomposition of the free space into trapezoids labeled with the same $c$-distance from $s$, and by running our algorithm without any modifications we discover the trapezoids in the $c$-map one-by-one (i.e., label-by-label), until all of the free space is lit. Denoting by $N$ the maximum complexity of the $c$-map (i.e., the number of trapezoids in the $c$-map) over $c\in C$, we thus obtain:
\begin{cor}An $O(CN)$-size data structure can be constructed to answer $C$-oriented link distance queries in $O(C\log N)$ time; a minimum-link path can be output in time proportional to the distance. The preprocessing time and space are $O(C^2N\log N)$ and $O(CN)$.\end{cor}
The above result is not entirely satisfactory because if $P$ is not $C$-oriented, $N$ may be unbounded and even infinite -- see, e.g., Fig.~\ref{fig:noapx} (so one would have to assume a model of computation in which operations on large numbers can still be carried out in constant time). This is due to the possibility that a trapezoid from $D_c^0$---the initial trapezoidal decomposition---may get split into an unbounded number of subtrapezoids in the final $c$-map. When can this happen? If there exists a $c$-oriented line, for some $c\in C$, that intersects both bases of the trapezoid, then the link distances to different points inside the trapezoid differ only by a constant, and, thus, all of it will be lit during $O(1)$ consecutive BFS steps---without any modifications of the algorithm. Hence, to avoid the dependence on $N$, we need to modify the algorithm only to handle trapezoids whose bases cannot be straddled by a $C$-oriented line. We determine such trapezoids in $D_c^0$ and mark them as \e{problematic}.

We will declare ``deep'' portions of problematic trapezoids as separate cells in the link distance map. The path to a query point $q$ inside such a cell consists of 2 parts: a path from $s$ to enter the trapezoid, and a ``zigzag'' of extreme orientations to $q$. The number of links in the first part is given by the usual link distance map, and the number of links in the second part can be determined in constant time (assuming constant-time floor function) because of the regular pattern of the path---it bounces off of the sides of the trapezoid until reaching the query point. Overall, we obtain

    \fi 
    
\begin{theorem}\label{thm:arb}An $O(Cn)$-size data structure can be constructed to answer $C$-oriented link distance queries in \e{arbitrary} domains in $O(C\log n)$ time; a minimum-link path can be output in time proportional to the distance. The preprocessing time and space are $O(C^2n\log n)$ and $O(Cn)$.\end{theorem}

    \ifarxiv
The proof of the theorem can be found in Appendix~\ref{app:problematic}.
    \fi

\subsubsection{Approximate $C$-oriented paths}\label{sec:coriApx}
    \ifcgta
\else

The $C$-oriented link distance can be 2-approximated by requiring that every second link of the path is horizontal. To find a minimum-link $C$-oriented path with this requirement one can do the BFS in the trapezoids intersection graph with one modification: instead of checking for intersection between trapezoids for all pairs of orientations, only check for intersections between the horizontal trapezoids and the other ones. Such a modification decreases the running time of our algorithm to $O(Cn\log n)$; the space remains $O(Cn)$, since the $C$ trapezoidations of $P$ are still constructed.

To reduce the space to $O(n)$ we do only the horizontal trapezoidation, and go back to the rectilinear-case ideas of Das and Narasimhan \cite{dn} who do the labeling of horizontal trapezoids \e{without} using the vertical ones (in our case, we label the horizontal trapezoids without using \e{any} other ones). In particular, without the other trapezoidations we cannot use our planting (there is simply nothing to plant!) to initialize the sweep. Thus we do both the UpSweep and the DownSweep as in \cite{dn}, starting from trapezoids labeled on the previous step.

Overall, we obtain

    \fi
\begin{theorem}An $O(n)$-size data structure can be constructed to answer $C$-oriented link distance queries to within a multiplicative error of 2; that is, if the minimum number of links in a $C$-oriented path from $s$ to the query point $q$ is $k$, the data structure will report a number $l\le2k$. The query time is $O(\log n)$, and a $C$-oriented $l$-link $s\-q$ path can be output in additional $O(l)$ time. The preprocessing time and space are $O(Cn\log n)$ and $O(n)$.\end{theorem}

    \ifarxiv
The proof of the theorem can be found in Appendix~\ref{app:coriApx}.
    \fi

\subsubsection{Approximating paths with ``robust'' edges}\label{sec:robust}We would like to use $C$-oriented paths to approximate the link distance of paths with unrestricted orientations. Unfortunately, as Fig.~\ref{fig:noapx} (cf.\ \cite[Fig.~1]{aos}) illustrates, the number of links in a minimum-link $C$-oriented path may be much higher than that of a path with unrestricted orientations: there are geometric configurations in which an unrestricted path of one link must be replaced by many, even infinitely many, \cori links. Even if we restrict the obstacle boundaries to use the same $C$ orientations (i.e., we work in a $C$-oriented domain), Fig.~\ref{fig:stillnoapx} illustrates that an unrestricted path of one link may require $\Omega(n)$ $C$-oriented links. To use $C$-oriented paths to approximate link distances, we define a model of paths that we call ``robust'' paths.\footnote{Note the contrast with \e{length} approximation: in terms of (Euclidean) length, $C$-oriented paths approximate general paths with additive relative error of $O(1/C^2)$, irrespective of whether the domain is $C$-oriented or not.}
%
\begin{figure}
\hfil
\begin{minipage}[c]{0.4\columnwidth}
\centering\includegraphics{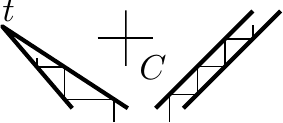}\caption{$C$-oriented min-link paths fail to approximate min-link paths with unrestricted orientations. Left: $t$ is not reached with finitely many links. Right: The number of $C$-oriented links required to pass through a corridor may depend on its width rather than on its complexity.}\label{fig:noapx}
\end{minipage}
\hfil
\begin{minipage}[c]{0.4\columnwidth}
\centering\includegraphics{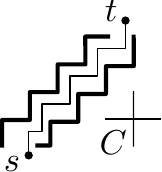}\caption{Even when the domain is $C$-oriented, a path having $O(1)$ links of arbitrary orientation can require $\Omega(n)$ $C$-oriented links.}\label{fig:stillnoapx}
\end{minipage}
\end{figure}

One feature is common to examples that show poor approximation power of $C$-oriented paths (e.g., Figs.~\ref{fig:noapx},\ref{fig:stillnoapx}): when an unrestricted path can have substantially fewer links than a $C$-oriented path, the set of orientations of a critical link of the unrestricted path is very narrow. This is also the case in instances of the minimum-link problem used in showing 3SUM-hardness (Fig.~\ref{fig:reduction})---an optimal path must ``shoot'' very precisely through free space. That is, a small deviation in the orientation of an edge in a minimum-link path renders the edge---and hence the whole path---infeasible.

We envision that such non-robustness of minimum-link paths with unrestricted orientations is highly undesirable in applications: if a path is to be followed by a robot, one would not want the robot to hit an obstacle due to a small deviation in the steering direction; if the path edges are communication links, the communication should not be interrupted due to a small error in the direction; etc. We thus quantify edge robustness as follows: for $\varphi\ge0$, an edge $pq$ is \e{$\varphi$-robust} if the isosceles triangle with altitude $pq$ and angle $2\varphi$ at the apex $p$ does not intersect obstacles (Fig.~\ref{fig:robust}). A path is $\varphi$-robust if all of its edges are $\varphi$-robust. (We use triangles centered on edges instead of circular sectors for technical reasons.) The definition is directional -- $pq$ can be $\varphi$-robust while $qp$ is not; one can modify the definition to make it symmetric. Note that according to our definition, robustness increases with $\varphi$; 0-robust is ``not at all robust''. Note also that adding vertices in the middle of a path's edges gives a new path whose robustness is at least that of the original (Fig.~\ref{fig:robust}, right). Intuitively this is fair: say, instead of sending a robot in a single command through a narrow corridor, with a single link, one may guide the robot through the corridor more carefully by adding extra reference points (path vertices) along the way. 

We postulate that for $C$-oriented paths there is no issue of ``wiggling'' of edges directions. That is, unlike with the unrestricted-orientations paths, if a \cori path is computed, it can actually be used with no fear of it becoming infeasible due to an error in edge orientation. The justification of such an assumption is twofold. First, from the theoretical point of view $C$-oriented paths are ``discrete'' and allow for no continuous change of edge orientation. Second, we are inspired by real-world mechanisms design: to set the orientation for an unrestricted path, one may need to turn a knob/throttle---this may not be easy to do with perfect precision. On the other hand, for $C$-oriented paths, the robot's wheels turning angle or the communication direction can be set with essentially no error. This can be achieved with any mechanism in which the shaft is connected to a gear that may rest against a notch---the gear is in a static equilibrium only when turned by a multiple of its angle of action. An example of such a mechanism is a ratchet (Fig.~\ref{fig:ratchet});
%
many others exist \cite{mechanicsBook}---facing crown gears, socket on a bolt, external gear fitting inside an external gear, etc.

Even though in general $C$-oriented paths do not approximate paths with unrestricted orientations at all (as seen in Figs.~\ref{fig:noapx},\ref{fig:stillnoapx}), as the next lemma shows, $C$-oriented paths approximate \e{robust} paths well. Specifically, let $\varphi = 180/i$, for some integer $i$, and let $C_\varphi$ be the set of orientations evenly spaced along the unit circle with angle $\varphi$ between consecutive orientations.
\begin{lemma}\label{lem:apx}If there exists a $k$-link $\varphi$-robust $s\-t$ path $\pi$, then there exists a $(k+1)$-link $C_\varphi$-oriented $s\-t$ path $\pi'$.\end{lemma}
\begin{proof}For an edge $p_ip_{i+1}$ of a $k$-link path $\pi=(s,p_1,p_2,\ldots,p_{k-1},p_k=t)$ let $T_{p_i}$ be the isosceles triangle with height $p_ip_{i+1}$ and angle $2\varphi$ at $p_i$. Refer to Fig.~\ref{fig:apx}.  Since $\pi$ is $\varphi$-robust, the triangle is obstacle-free and contains a $C_\varphi$-oriented segment connecting $p_i$ to the base of the triangle on each side of $p_ip_{i+1}$. Denote these {\em snapped} segments by $p_ip_i^+,p_ip_i^-$. Let $q_k$ be the point within $T_{p_{k-1}}$ where the line through $t$ that is parallel to $p_{k-1}p_{k-1}^-$ intersects $p_{k-1}p_{k-1}^+$.  Now, the line through $p_{k-1}$ and $q_k$ must cross one of the two snapped segments through $p_{k-2}$ ($p_{k-2}p_{k-2}^+$ or $p_{k-2}p_{k-2}^-$). Let $q_{k-1}$ be the point of crossing.  (In fact, the line through $p_{k-1}$ and $q_k$ must intersect {\em exactly one} of the two snapped segments, unless it is perpendicular to $p_{k-2}p_{k-1}$ (meaning it contains the base, $p_{k-2}^-p_{k-2}^+$, of $T_{k-2}$), in which case we select $q_{k-1}$ to be either point $p_{k-2}^-$ or $p_{k-2}^+$.)  Note that the segment $q_{k-1}q_k$ lies inside $T_{k-2}\cup T_{k-1}$ and is, therefore, within $P$.  (Segment $q_{k-1}q_k$ lies inside $T_{k-1}$ if $\angle p_{k-2}p_{k-1}q_k<90^\textrm{o}$; otherwise, $q_{k-1}q_k=q_{k-1}p_{k-1}\cup p_{k-1}q_k$, with $q_{k-1}p_{k-1}\subset T_{k-2}$ and $p_{k-1}q_k\subset T_{k-1}$.)  Similarly, we define $q_i$, for $i=k-2, k-3,\ldots, 1$, to be the point where the line through $p_i$ and $q_{i+1}$ crosses one of the two snapped segments through $p_{i-1}$. Note that each segment $q_iq_{i+1} \subset T_{i-1}\cup T_i$, so that the {\em snapped path} $\pi'=(s,q_1,q_2,\ldots,q_k,t)$ lies within $P$ and has each of its $k+1$ links $C$-oriented.
\end{proof}
Note that the lemma holds in an arbitrary domain (no $C$-orientedness of $P$ is used in the proof, only $\varphi$-robustness of $\pi$). Thus, along with the results from Section~\ref{sec:coriinarb}, the lemma implies
\begin{cor}An $O(\frac{n}\varphi)$-size data structure can be constructed to answer $\varphi$-robust link distance queries to within a one-sided additive error of 1. That is, if the minimum number of links in a $\varphi$-robust path from $s$ to the query point $q$ is $k$, the data structure will report a number $l\le k+1$. The query time is $O(\frac{\log n}\varphi)$, and a $C_\varphi$-oriented $l$-link $s\-q$ path can be output in additional $O(l)$ time. The preprocessing time and space are $O(\frac{n}{\varphi^2}\log n)$ and $O(\frac{n}\varphi)$.\end{cor}
\begin{figure}
\begin{minipage}[c]{0.4\columnwidth}
\centering\includegraphics{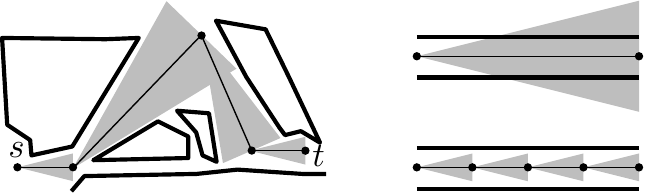}\caption{Left: A $\varphi$-robust 4-link path; $2\varphi$ is the subtended angle of the shaded wedges. A 2-link $s\textrm-t$ path exists, but it is not robust. Right: Adding vertices to increase the robustness.}\label{fig:robust}
\vspace{10pt}
\centering\includegraphics[height=1in]{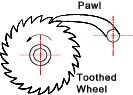}\caption{Ratchet (yankeewombat.com).}\label{fig:ratchet}
\end{minipage}
\hfill
\begin{minipage}[c]{0.5\columnwidth}
\centering\includegraphics[scale=0.45]{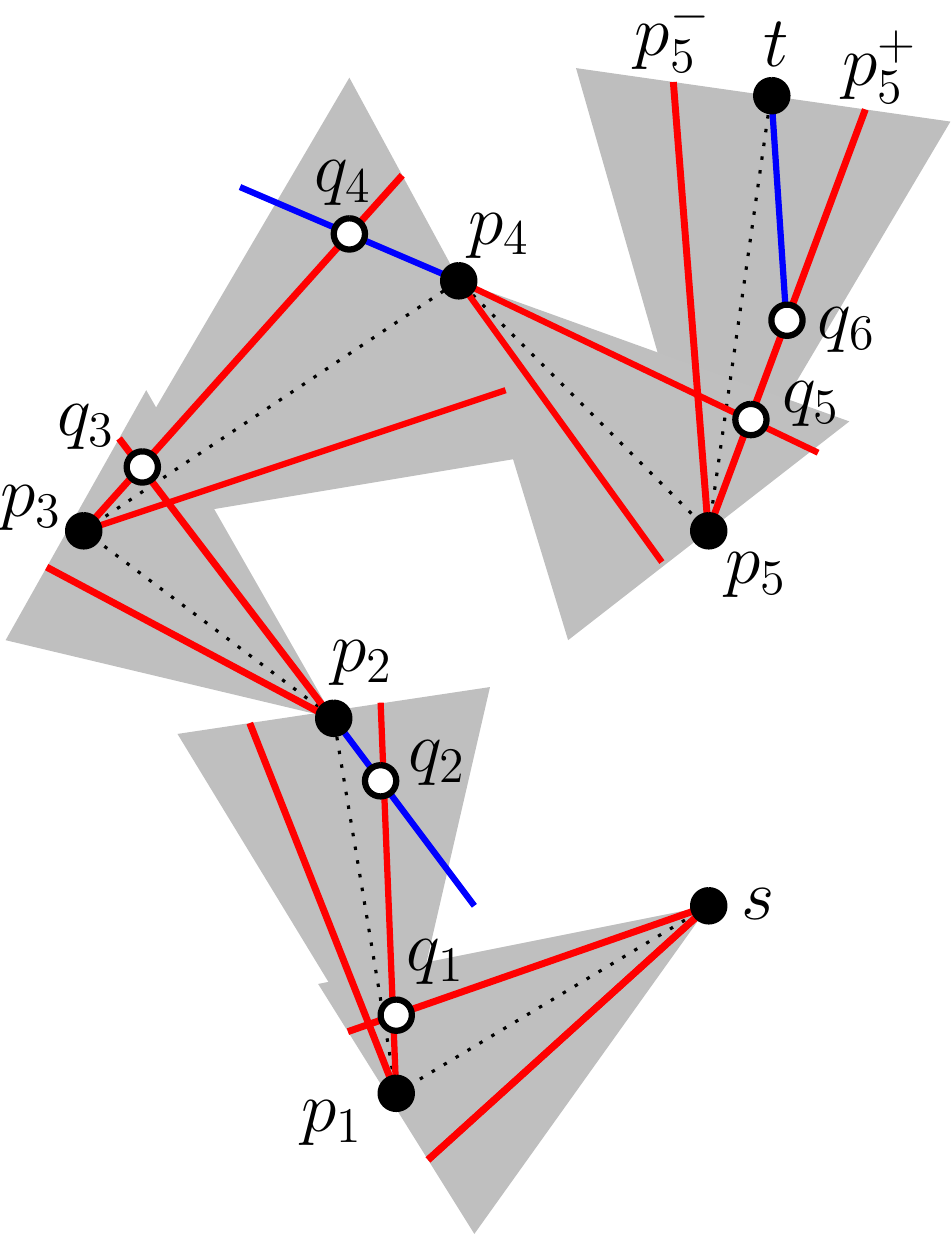}\caption{Converting an arbitrarily oriented robust path $\pi$ (dotted) into a $C_\varphi$-oriented path $\pi'$ that follows snapped segments (red). Vertices of $\pi'$ are the hollow circles. There is one vertex per link of~$\pi$.}\label{fig:apx}
\end{minipage}
\end{figure}

\section{Conclusion}We presented a collection of results related to finding minimum-link paths in polygonal domains. Many open questions remain, including:
\begin{description}
\item Is it 3SUM-hard to compute minimum-link rectilinear paths in 3D, or can one obtain a subquadratic-time algorithm? A nearly-quadratic-time algorithm for 3D was given in~\cite{wads11}.
\item Can one reduce the gap between 2 and $O(\sqrt{h})$ for approximability of minimum-link paths in polygonal domains with holes?
\end{description}
\paragraph{Acknowledgments}We are grateful to the three anonymous reviewers for their many comments that helped improving the presentation of the paper.  We thank Estie Arkin for suggesting the 3SUM-hardness construction, and we thank Haitao Wang for discussions on several aspects of the paper, including the remark that one can test in quadratic time if the $s\textrm-t$ link distance is at most 3.  J. Mitchell is partially supported by the National Science Foundation (CCF-1018388), Metron Aviation, and NASA Ames. V. Polishchuk is funded by the Academy of Finland grant 1138520. M. Sysikaski is funded by the Research Funds of the University of Helsinki.
\bibliographystyle{abbrv}\bibliography{ml}

\ifarxiv
    \clearpage
\fi

\section*{Appendix}
\appendix\section{A note on a claim in \cite{ylwBends}}\label{app:Lee}The abstract of \cite{ylwBends} announces an optimal $O(n\log n)$-time $O(n)$-space algorithm to find a minimum-link path (a minimum-\e{bend} path, or MBP, in terminology of \cite{ylwBends}): ``optimal $\theta(e \log e)$ time algorithms are presented to find the shortest path and the minimum-bend path using linear space''. We did not find the MBP considered anywhere in the paper, though. Moreover, the penultimate paragraph of the Introduction (p.~712) states that ``problems MBP and yD-SP have been solved optimally in [5], [11], [24], [25]''. The references [5], [11], [24], [25] from \cite{ylwBends} are our \cite{dn,iaFocs,os,ohtsuki}. However, Imai and Asano \cite{ia1,ia2,iaFocs} claimed $O(n\log n)$ time \e{and} space for their method. Also, Ohtsuki's algorithm \cite{ohtsuki} runs in $O(n\log^2n)$ time and $O(n)$ space (as acknowledged in the second paragraph of the Introduction \cite[p.~711]{ylwBends}). It is possible that, at the time of writing \cite{ylwBends}, only the time bound was considered, and not the space bound, in describing the ``optimality'' of the algorithm. Unfortunately, too much time passed to resolve the possible confusion now \cite{LeePersonal}.

   \ifarxiv

\section{Handling problematic trapezoids}\label{app:problematic}Recall that a trapezoid is \e{problematic} if its bases cannot be straddled by a $C$-oriented line.

\begin{wrapfigure}{l}{.3\columnwidth}\centering\includegraphics{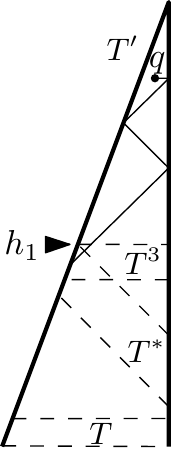}\caption{A zigzag through $T'$. The labels of $T,T^*,T^3$ are $k,k+2,k+3$, respectively.}\label{fig:problematic}\end{wrapfigure}
Let $T'$ be a problematic $c$-trapezoid. As in Section~\ref{subsec:cori}, assume the orientation $c\in C$ is horizontal and implicit, and $C^*=C\setminus c$. We first describe how to handle $T'$ when it is a triangle, i.e., when one of its bases has length 0 (Fig.~\ref{fig:problematic}). Without loss of generality, the 0-length base is the upper base of $T'$. Let $c_1,c_2\in C$ be the orientations closest to those of the sides of $T'$; we refer to orientations $c_1$ and $c_2$ as \e{extreme} orientations. We say that a path $\pi'$ from a point $q\in T'$ is a \e{zigzag} if it starts with a horizontal segment and its vertices bounce between the sides of $T'$ using extreme-orientation links 
(two zigzags start at $q$, one going left, the other going right). Let $S_1,S_2$ denote the sides of $T'$. Given two points $h_1\in S_1,h_2\in S_2$, the following query can be answered in constant time (assuming constant-time floor function): What is the minimum number of links in a zigzag from $q$ needed to reach $S_1$ below $h_1$ or to reach $S_2$ below $h_2$? Denote the answer to the query by $K(q,h_1,h_2)$. 

Let $k$ be the BFS step at which $T'$ is first partially lit. Continue the BFS for 3 more steps. Let $T^3\subset T'$ be the label-$(k+3)$ $c$-trapezoid, and let $h_1\in S_1,h_2\in S_2$ be the highest points on the sides of $T'$ reached by the extreme-orientations label-$(k+2)$ trapezoids. We declare the part of $T'$ above $T^3$ to be a separate single cell in the $c$-map. The $c$-distance to a query point $q$ in the cell is given by $k+2+K(q,h_1,h_2)$.

To justify the correctness of our approach, recall from Section~\ref{subsec:int} that only flush-lighting leads to trapezoid splitting (straddle-lighting does not). This means that the subtrapezoid $T\subset T'$ lit at step $k$ was flush-lit. Similarly, every subtrapezoid $T''\subset T'$ whose label $k''$ is larger than $k$ is flush-lit. Consider a subtrapezoid $T''\subset T'$ with label $k''\ge k+3$. Let $T^*$ be the $c^*$-trapezoid with label $k''-1$ that (flush-)lit $T''$. It appears that $T^*$ was flush-lit itself, for otherwise its side must have been supported outside $T'$, but then $T^*$ must have intersected with the label-$k$ trapezoid $T$---a contradiction to $k''-1\ge k+2$.

The crucial observation about flush-lighting is that a minimum-link path to a point $q$ inside a flush-lit trapezoid without loss of generality bounces off of the trapezoid side (i.e., the path's last turn before $q$ without loss of generality belongs to a side of the trapezoid). Hence, a minimum-link $c$-path $\pi'$ from a point $q\in T''$ without loss of generality bounces off of the sides of $T'$ (at least) until reaching a label-$(k+2)$ trapezoid. By a local modification argument, it does not hurt to make all non-horizontal, bouncing links of $\pi'$ extreme. Thus, without loss of generality, the part of $\pi'$ reaching an extreme label-$(k+2)$ trapezoid is a zigzag, which proves the formula for the link distance to~$q$.


We now describe the modifications of the procedure to handle a problematic trapezoid $T'$ whose bases are each of positive length. Exactly as above, we find $h_1\in S_1$ and $h_2\in S_2$, the highest points on the sides of $T'$ reached by the extreme-orientations label-$(k+2)$ trapezoids. We then ``propagate'' zigzags from $h_1$ and $h_2$ up to the other end of $T'$. Specifically, let $T^*$ be the $c^*$-trapezoid, supported by a side of $T'$, that penetrates deepest into $T'$ from its upper base ($T^*$ can be read from the support list of the edge supporting the side). In constant time we can determine the minimum number $K^*$ of links necessary to reach $T^*$ from $h_1$ or $h_2$. We assign \e{temporary} label $k+2+K^*$ to $T^*$. If $T^*$ is not lit by BFS step $k+2+K^*$, the label becomes permanent, and $T^*$ is inserted into $S_{c_i}^{k+2+K^*}$. Otherwise, if $T^*$ is lit before the step $k+2+K^*$, we propagate the zigzag path from it down through $T'$. In constant time we determine where the path meets with zigzags from $h_1,h_2$ and establish a ``bisector'' trapezoid $T_b\subset T'$ in the $c$-map; the bisector can be reached through either base of $T'$. The portions of $T'$ below and above the bisector become separate cells in the map, just as in the case of a degenerate, triangular trapezoid described above.

In comparison with Section~\ref{subsec:cori}, we do $O(C)$ additional work per trapezoid in $D^0$. Thus our time and space bounds of Theorem~\ref{thm:cori} carry over to the case of \cori paths in arbitrarily oriented domains.

\section{Approximating \cori paths}\label{app:coriApx}We are 2-approximating the $C$-oriented link distance by requiring that every second link of the path is horizontal. To find a minimum-link $C$-oriented path with this requirement, it is enough to do only the horizontal trapezoidation, and do a BFS in it (again, starting from the maximal free-space horizontal segment though $s$): label with $k+2$ the unlabeled trapezoids that are seen, in at least one direction $c^*\in C$, from a label-$k$ trapezoid. This is \e{exactly} the problem solved by Das and Narasimhan \cite{dn} (for the case when $c^*$ is vertical). Thus, for each orientation $c^*$ we can march through the decomposition with a $c^*$-UpSweep and a $c^*$-DownSweep of \cite{dn}: just as \cite{dn}, we do not need the $c^*$-trapezoidation for that. The only difference is that when $c^*$  is not vertical, some trapezoids may be lit only partially, but this is easy to keep track of by maintaining, for each side of every trapezoid, the maximum and minimum height reached by the $c^*$-rays coming from label-$k$ trapezoids.

To give more details, we first recap the sweep algorithm of Das and Narasimhan \cite{dn}. Recall that the goal of the UpSweep is as follows: Given a set $S$ of trapezoids (labeled $k$), label with $k+2$ the trapezoids that are seen by looking up from $S$. The sweepline starts at $-\infty$, with empty status, and moves upward. The sweepline status is the set of points on the line that are seen by looking up from $S$. Thus, the status is a set of disjoint intervals, which can be kept in a simple structure like the one in Section~\ref{sec:rect} (we see no need to use a complicated structure to maintain some ``fronts'' of merging and splitting ``windows'' as proposed in \cite{dn}).

The events are bases of trapezoids. The events are ordered by height. In case of ties, upper bases are processed before lower bases. Initially, the queue contains the trapezoids~$S$.

Processing of the lower base of a trapezoid $T\in S$ is simple: the base is added to the sweepline status. Processing of the upper base of a trapezoid $T\in S$ involves the following: (1)~the part of the base that rests against an obstacle (if any) is removed from the status; (2)~bases of the upper neighbors of $T$ enter the queue.

Processing of the lower base of a trapezoid $T'\notin S$ is as follows: if the base does not overlap with the sweepline status, the upper base of $T'$ is removed from the queue; otherwise, $T'$ is labeled $k+2$. Processing of the upper base of a trapezoid $T'\notin S$ happens, therefore, only if $T'$ is labeled $k+2$ and involves the same steps as processing the upper base of a trapezoid $T\in S$: (1)~the part of the base that rests against an obstacle (if any) is removed from the status; and, (2)~bases of the upper neighbors of $T$ enter the queue.

The DownSweep is analogous.
\begin{remark}Not surprisingly, there are clear parallels between (the recap of) Das and Narasimhan's algorithm given here and our planting-based, upsweep-only algorithm from Section~\ref{sec:rect}. Both algorithms use the same idea of labeling the trapezoids step-by-step. Moreover, technically, the sweepline status and the event queue are the same in both algorithms. The only formal difference is how the event queue is initialized: with the aid of planting (as in Section~\ref{sec:rect}) or just by $S$ (here). The conceptual difference is whether the vertical trapezoidation is used explicitly (as in Section~\ref{sec:rect}) or not. We extended both ideas to the general case of $C\ge2$. The former appeared more useful for computing exact minimum-link $C$-oriented paths (Section~\ref{subsec:cori}); the latter is useful in achieving our goal here of 2-approximating a $C$-oriented path in $O(n)$ space, without building the other $C-1$ trapezoidations.\end{remark}
In the general case (when the domain is not rectilinear and $c^*$ is not vertical) the goal of the UpSweep is to label with $k+2$ the trapezoids that are seen by looking from $S$ in direction $c^*$. Analogously to the rectilinear case, the sweepline status is the set of points on the sweepline that are seen by looking up in orientation $c^*$. The issue now is that if a status interval $I$ touches a side of a trapezoid $T$, then the status changes continuously (Fig.~\ref{fig:upsweep}). Nevertheless, since $T$ is obstacle-free, $I$ cannot intersect any other trapezoid inside $T$. Thus, as soon as the sweepline enters $T$, we clip $I$ to what it should be after the sweepline exits~$T$.
\begin{figure}\centering\includegraphics[scale=.9]{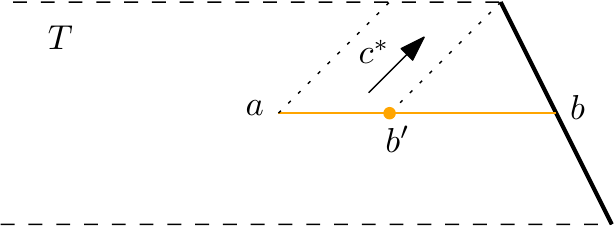}\caption{As the sweepline moves up, the sweepline status interval $ab$ changes continuously, as $b$ slides along the side of $T$. However, inside $T$ the interval does not intersect any other trapezoid. Thus, as soon as the sweepline enters $T$, $ab$ can be clipped to $ab'$; there will be no undetected intersections (no false-negatives).}\label{fig:upsweep}\end{figure}

Another change we need for the lower-base events is due to the fact that a trapezoid may get only partially lit; thus, we need to find the highest point of $T$ touched by the intervals. This is easy to do by looking at the intervals that we clip away. After finishing all of the sweeps we split partially lit trapezoids as in Section~\ref{subsec:cori}.

The described clipping of the sweepline status on a lower-base event and the trapezoid splitting are the only differences from the rectilinear case. Upper-base events are handled exactly as before.

    \fi

\end{document}